\documentclass[nofootinbib,notitlepage,twocolumn]{revtex4-1} 

\pdfoutput=1
\usepackage{hyperref}
\usepackage{dsfont}
\usepackage{tikz}
\usetikzlibrary{decorations.pathreplacing}

\usepackage{braket}
\usepackage{mathrsfs}
\usepackage{graphicx}
\usepackage{amsthm,amsmath,amssymb,bbm,amsfonts}
\usepackage[bb=boondox]{mathalfa}

\newcommand{\ii}{\mathrm{i}}
\newcommand{\ie}{{i.e.}\ }

\DeclareMathOperator{\diag}{diag}
\DeclareMathOperator{\Vol}{Vol}

\newtheorem{theorem}{Theorem}
\newtheorem{definition}{Definition}

\begin{document}

\title{Entanglement production in the dynamical Casimir effect at parametric resonance}

\author{Ivan Romualdo}
\email{ivanromualdo@usp.br}
\affiliation{Departamento de F\'isica - ICEx, Universidade Federal de Minas Gerais, CP 702, 30161-970, Belo Horizonte - MG, Brazil}
\affiliation{Instituto de F\'isica-Universidade de S\~ao Paulo, CP 66318, 05315970- S\~ao Paulo-SP, Brazil}
\author{Lucas Hackl}
\email{lucas.hackl@mpq.mpg.de}
\affiliation{Max Planck Institute of Quantum Optics, Hans-Kopfermann-Str. 1, D-85748 Garching, Germany}
\affiliation{Munich Center for Quantum Science and Technology, Schellingstra{\ss}e 4, D-80799 M\"{u}nchen, Germany}
\author{Nelson Yokomizo}
\email{yokomizo@fisica.ufmg.br}
\affiliation{Departamento de F\'isica - ICEx, Universidade Federal de Minas Gerais, CP 702, 30161-970, Belo Horizonte - MG, Brazil}

\begin{abstract}
The particles produced from the vacuum in the dynamical Casimir effect are highly entangled. In order to quantify the correlations generated by the process of vacuum decay induced by moving mirrors, we study the entanglement evolution in the dynamical Casimir effect by computing the time-dependent R\'enyi and von Neumann entanglement entropy analytically in arbitrary dimensions. We consider the system at parametric resonance, where the effect is enhanced. We find that, in $(1+1)$ dimensions, the entropies grow logarithmically for large times, $S_A(\tau)\sim\frac{1}{2}\log(\tau)$, while in higher dimensions $(n+1)$ the growth is linear, $S_A(t)\sim \lambda\,\tau$ where $\lambda$ can be identified with the Lyapunov exponent of a classical instability in the system. In $(1+1)$ dimensions, strong interactions among field modes prevent the parametric resonance from manifesting as a Lyapunov instability, leading to a sublinear entropy growth associated with a constant rate of particle production in the resonant mode. Interestingly, the logarithmic growth comes with a pre-factor with $1/2$ which cannot occur in time-periodic systems with finitely many degrees of freedom and is thus a special property of bosonic field theories.
\end{abstract}

\maketitle

\section{Introduction}
The dynamical Casimir effect (DCE) describes the creation of particles from the vacuum by moving mirrors \cite{moore1970quantum,Dodonov:2010zza}.  Many of its basic features, as particle production and the dynamics of quantum fields with time-dependent boundary conditions, closely resemble analogous effects in sophisticated phenomena as the Hawking effect, the Unruh effect, or particle production in an expanding Universe. As a result, the DCE has been extensively investigated as a simple model to explore physical aspects of these more complicated systems in a more manageable context \cite{Davies:1977yv,Carlitz:1986nh, Brevik:2000zb, dodonov2001nonstationary,Lock:2016rmg}. In addition, the DCE has been observed experimentally in a superconducting circuit \cite{wilson2011observation}, opening an avenue for the investigation of processes of particle production from the vacuum in tabletop experiments\footnote{In the experiment \cite{wilson2011observation}, the idealized time-dependent boundary conditions imposed by the moving mirror are implemented in practice by changing the properties of the superconducting circuit by the application of a magnetic flux.}.

\begin{figure}
    \centering
    \begin{tikzpicture}
    \draw[dashed] (5.5,0) -- (5.5,1) -- (6.5,1) -- (6.5,0);
    \draw[thick,->] (0,-.7) -- (0,3.5) node[left,xshift=-1mm]{$t$};
    \draw[thick,->] (-.7,0) -- (7,0) node[below,yshift=-1mm]{$x$};
    \draw[very thick, blue] (0,-.7) -- (0,2.5) node[above,fill=white]{left mirror};
    \draw[dashed,thick] (6,0) -- (6,2.5);
    \draw[domain=0:2,smooth,variable=\t,red,very thick] plot ({6+.5*sin(deg(2*3.1415*\t))},{\t});
    \draw[very thick,red] (6,-.7) -- (6,0) (6,2) -- (6,2.5) node[above]{right mirror};
    \draw [decorate,decoration={brace,amplitude=8pt}] (6,0) -- node [yshift=-6mm] {$L_1$} (0,0);
    \draw [decorate,decoration={brace,amplitude=3pt}] (6.5,0) -- node [yshift=-3mm,xshift=1mm] {$\epsilon L_1$} (6,0);
    \draw [decorate,decoration={brace,amplitude=8pt}] (6.5,1) -- node [xshift=6mm] {$T$} (6.5,0);
    \draw[very thick,brown] (0,1.1) -- (6.29389,1.1);
    \draw (3,1.1) node[above]{$L(t)=L_1[1+\epsilon\,\sin(2\pi\,t/T)]$};
    \end{tikzpicture}
    \caption{Basic scenario of the dynamical Casimir effect in $(1+1)$ dimensions. We have a standing mirror at the origin ($x=0$) and a moving mirror ($x=L(t)$). The distance between the mirrors oscillates as $L(t)=L_1[1+\epsilon \sin{(2\pi t/T)}]$ with period $T$. Note that the mirror only moves during a finite time interval.}
    \label{fig:dynamical-casimir}
\end{figure}
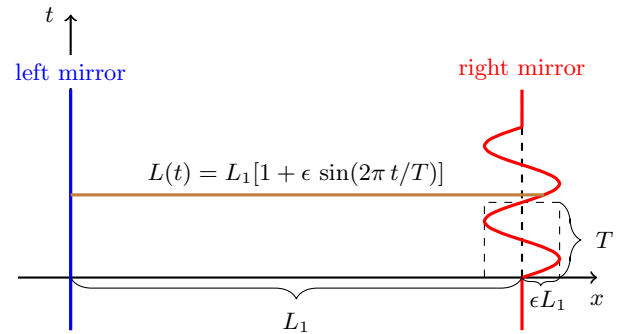

Because the DCE effects are usually small in laboratory conditions, it is interesting to analyze it in a condition of resonance \cite{dodonov1992long,Dodonov:1996zz,lambrecht1996motion,plunien2000dynamical}. In this work,  we focus on the case where a mode of the field is at parametric resonance with the mirror oscillation, following \cite{Dodonov:1996zz}. The general  setup of the DCE at parametric resonance is depicted in Fig.~\ref{fig:dynamical-casimir}. A one dimensional cavity is bounded by two mirrors, one of which is static, while the other is allowed to move. If the moving mirror is set to oscillate harmonically with twice the frequency of some mode of the field, such a mode becomes resonant. This leads to a strong enhancement of the effect, with particles being continuously produced from the vacuum, providing a natural scenario for eventual further experimental explorations of the DCE \cite{Dodonov:2010zza}.

Several works deal with the most immediate aspects of the phenomenon, such as the time evolution of particle production and the Casimir force generated between the plates (for a review, see section II of \cite{dodonov2001nonstationary}). However, a highly unexplored direction is the dynamics of the quantum correlations that emerge due to particle creation and the mixing of field modes induced by the moving mirror. In this work, we contribute to fill this gap by analitically determining the time evolution of entanglement measures for the DCE at parametric resonance.

We compute the R\'{e}nyi and the entanglement entropy production from an initial vacuum state for the DCE at parametric resonance in arbitrary dimensions. We consider a massless scalar field in a hyperrectangular cavity where the field is required to vanish on the boundaries. For $t\leq0$, the cavity is kept fixed. For $t>0$, we allow the length of the cavity to vary in one direction for a finite amount of time by keeping one mirror fixed and setting the other to oscillate according to 
\begin{equation}
    L(t)=L_1[1+\epsilon\sin{(2\pi t/T)}].
\end{equation}
After a finite number of oscillations, we then let the mirror return to its starting position, and the cavity is kept fixed again. The frequency of the mirror oscillation is set to twice that of some cavity mode, giving rise to parametric resonance. We select the lowest energy mode to be the resonant, since this is the most likely scenario to be reproduced in a laboratory. The time evolution of particle production in such conditions was fully analyzed in \cite{Dodonov:1996zz}, where, under reasonable approximations, the Bogoliubov transformation between the initial and final states was computed analytically for all times in both $(1+1)$ and $(3+1)$ dimensions.

In order to compute the entropies, we apply recently developed symplectic techniques for the description of the dynamics of entanglement for Gaussian states in linear bosonic systems \cite{Bianchi:2015fra,Bianchi:2017kgb}. In this framework, the R\'{e}nyi entropy is first computed as a volume associated with the subsystem in phase space. Then, according to a theorem derived in \cite{Bianchi:2017kgb}, the asymptotic behavior of the R\'{e}nyi and entanglement entropies will converge. Moreover, in the presence of Lyapunov instabilities, the asymptotic production of entanglement entropy is completely determined by the Lyapunov exponents. We show that in $(1+1)$ dimensions the entanglement entropy grows logarithmically for large times. In higher dimensions, the resonance is associated with a Lyapunov instability, leading to an asymptotic exponential particle production and linear growth of the entanglement entropy, with a production rate equal to the Lyapunov exponent of the system.

This paper is structured as follows. In Section \ref{sec:Ham.DCE}, we review the Hamiltonian formalism of the dynamical Casimir effect and the analytical solution for the case of  parametric resonance obtained in \cite{Dodonov:1996zz}. In Section \ref{sec:Ent.G. states}, we discuss the necessary background on symplectic techniques for the dynamics of entanglement of Gaussian states. In Section~\ref{sec:Ent.Prod}, we compute analytically the time evolution of the R\'enyi and entanglement entropies in the dynamical Casimir effect at parametric resonance. Finally, we summarize and discuss our results in Section~\ref{sec:discussion}. 

\section{Hamiltonian formalism for the dynamical Casimir effect}\label{sec:Ham.DCE}

\subsection{The model}
Let $\phi(t,x)$ be a neutral massless scalar field in $(d+1)$-dimensions described by the Lagrangian
\begin{equation}
    L = -\frac{1}{2} \int d^n x  \, (\partial^\mu \phi) \, (\partial_\mu \phi) \, ,
\end{equation}
where we adopt the mostly plus convention for the Minkowski metric, $\eta_{\mu \nu} = \diag(-,+,\dots,+)$. Requiring the action to be stationary, one finds that the field satisfies the Klein-Gordon equation
\begin{equation} \label{eq:KG}
    \frac{\partial^2 \phi}{\partial t^2} - \nabla^2 \phi = 0 \, .
\end{equation}
The associated momentum is
\begin{equation}
    \pi(t,x) = \partial_t \phi(t,x) \, ,
\end{equation}
and the Hamiltonian is
\begin{equation}
    H= \frac{1}{2} \int d^n x  \, \left[  \pi^2 + (\nabla \phi)^2 \right] \, .
\end{equation}

We are interested in the situation where the field is confined to a finite spatial volume $V(t)$. More specifically, we consider a field restricted to a time-dependent parallelepiped bounded in some direction $x^1$ to a finite interval $x^1 \in [0,L(t)]$, where the first boundary is fixed and the second boundary is allowed to move arbitrarily. In the remaining directions, the boundaries are fixed, so that $x^i \in [0,L_i]$, $i=2,\dots,d$, for $x \in V(t)$. Dirichlet boundary conditions are imposed at $\partial V(t)$:
\begin{equation} \label{eq:boundary-conditions}
    \phi(t,x)|_{\partial V(t)} = 0 \,.
\end{equation}
In particular,
\begin{equation}
    \phi(t,x^1=0) = \phi(t,x^1=L(t))=0 \, .
\end{equation}
If the field $\phi$ represents a component of the electromagnetic field, then these boundary conditions correspond to perfect mirrors, and the field is confined to a cavity bounded by perfectly reflecting plates. The boundary is assumed to move only for a finite period of time $T$, returning at the end to the starting position $L_1$:
\begin{equation}
    L(t) = L_1 \, , \quad \text{ for } t\leq 0 \text{ and } t \geq T \,.
\end{equation}
In the quantum theory, the motion of the mirrors will lead to particle production, and under this restriction the number of produced particles remains finite.

Let $\mathcal{S}$ be the space of solutions of the Klein-Gordon equation \eqref{eq:KG} with boundary conditions \eqref{eq:boundary-conditions}. This space is equipped with an invariant bilinear form
\begin{equation}
    (\phi,\psi) = i \int_{V(t)} d^n x \, \left[ \phi^* (\partial_t \psi) - (\partial_t \phi^*) \psi \right] \, ,
\end{equation}
which is not positive definite. The space of solutions can be decomposed into a direct sum of subspaces of positive and negative energy, $\mathcal{S}=\mathcal{S}^+ \oplus \mathcal{S}^-$. The decomposition is required to be such that (i) positive-energy solutions have positive squared norm $(\phi,\phi)$ and (ii) the complex conjugate of a positive-energy solution is a negative-energy solution, $\mathcal{S}^- = \overline{\mathcal{S}^+}$.
The decomposition is not unique, and distinct decompositions will correspond to distinct choices of vacuum in the quantum theory. The restriction of the invariant bilinear form to $\mathcal{S}^+$ defines an inner product in this subspace, turning it into a Hilbert space $\mathcal{H}$. We denote by $\{u_n(t,x)\}$ an orthonormal basis of $\mathcal{H}$. The solutions $u_n$ are positive-energy normal modes for the Klein-Gordon equation.

The canonical quantization of the scalar field $\phi(t,x)$ and its associated momentum $\pi(t,x)$ provides field operators satisfying the canonical comutation relations at equal times:
\begin{align}
    [\hat{\phi}(t,x),\hat{\pi}(t,x')] & = i \delta(x-x') \, , \\
    [\hat{\phi}(t,x),\hat{\phi}(t,x')] &= [\hat{\pi}(t,x),\hat{\pi}(t,x')] = 0 \, ,
\end{align}
and the equation of motion \eqref{eq:KG}. For each choice of normal modes $u_n$, the field operator can be expanded as
\begin{equation} \label{eq:field-expansion}
    \hat{\phi}(t,x) = \sum_n [u_n(t,x) \hat{a}_n + u^*_n(t,x) \hat{a}_n^\dagger] \, ,
\end{equation}
where the $\hat{a}_n, \hat{a}_n^\dagger$ are creation and annihilation operators satisfying the usual commutation relations:
\begin{equation}
    [\hat{a}_m,\hat{a}_n^\dagger] = \delta_{mn} \, , \quad [\hat{a}_m,\hat{a}_n] = [\hat{a}_m^\dagger,\hat{a}_n^\dagger]=0\, .
\end{equation}
The quantized canonical momentum operator is obtained by taking the time derivative of the field operator in the representation \eqref{eq:field-expansion}. Distinct choices of normal modes correspond to distinct representations of the quantum field. The vacuum state $\ket{0}$ is the state annihilated by all $\hat{a}_n$'s and depends on the choice of $\mathcal{H}=\mathcal{S}^+$ (although not on the specific basis $\{u_n\}$ chosen for $\mathcal{H}$). 

We introduce two special representations, called the \emph{in}- and \emph{out}-representations. For $t<0$, the Hamiltonian is time-independent, and one can find an orthonormal basis of positive energy $in$-solutions for which the time-dependence is factored out in this region:
\begin{equation}
    u_n^{(in)}(t,x) = \psi^{n}(t,x) \propto e^{- i \omega_n t} \, , \quad \text{for } t \leq 0 \, .
\end{equation}
Similarly, one can find an orthonormal basis of positive energy $out$-solutions satisfying
\begin{equation}
    u_k^{(out)}(t,x) = \psi_k(t,x) \propto e^{- i \omega_k t} \, , \quad \text{for } t\geq T \, .
\end{equation}
Both sets $\mathcal{B}_{in}=\{\psi^n,\psi^{n *}\}$ and $\mathcal{B}_{out}=\{\psi_k,\psi_k^*\}$ are linear bases of $\mathcal{S}$, but define distinct decompositions into positive and negative energy subspaces. The field operator can be expanded in either set of normal modes:
\begin{align} \label{eq:rep-out}
    \hat{\phi}(t,x) &= \sum_k [\psi_k(t,x) \hat{a}_k + \psi_k^*(t,x) \hat{a}_k^\dagger] \, , \\
    \label{eq:rep-in}
    &= \sum_n [\psi^n(t,x) \hat{b}_n + \psi^{n*}(t,x) \hat{b}_n^\dagger] \, ,
\end{align}
where the $\hat{b}_n,\hat{b}_n^\dagger$ are annihilation and creation operators in the $in$-representation, and the $\hat{a}_k,\hat{a}_k^\dagger$ are annihilation and creation operators in the $out$-representation. Expanding the $in$-modes in the basis of $out$-modes,
\begin{equation} \label{eq:psi-in-out}
    \psi^n = \sum_k (\alpha_{nk} \psi_k + \beta_{nk} \psi_k^*) \, ,
\end{equation}
and substituting into Eq.~\eqref{eq:rep-in}, we find
\begin{align} \label{eq:B-transf-complex}
    \hat{a}_k &= \sum_n (\alpha_{nk} \hat{b}_n + \beta_{nk}^* \hat{b}_n^\dagger) \, , \\
    \hat{a}_k^\dagger &= \sum_n (\alpha_{nk}^* \hat{b}_n^\dagger + \beta_{nk} \hat{b}_n) \, .
\end{align}
This is a Bogoliubov transformation describing the relation between the two representations. The constants $\alpha_{nk},\beta_{nk}$ are the Bogoliubov coefficients of the transformation. The transformation is encoded in a symplectic matrix:
\begin{equation} \label{eq:Bogoliubov-matrix-general}
    M = \begin{pmatrix}
         \alpha^T & \beta^\dagger   \\
         \beta^T  & \alpha^\dagger
    \end{pmatrix} \, .
\end{equation}

The evolution of any property of the quantum field due to the motion of the boundaries for $t \in [0,T]$ is completely characterized by the Bogoliubov coefficients. In particular, the process of particle creation from the initial vacuum has a simple description. Let the initial state for $t\leq 0$ be the $in$-vacuum annihilated by all $\hat{b}_n$:
\begin{equation}
    \hat{b}_n \ket{0,in} = 0 \, , \quad \forall n.
\end{equation}
Since we are in the Heisenberg representation, the state of the field remains unchanged. However, the number of particles $N_k$ in a mode $k$ observed after the motion of the boundary is described by the $out$-operators:
\begin{equation}
    \langle N_k \rangle = \bra{0,in} \hat{a}_k^\dagger \hat{a}_k \ket{0,in} = \sum_n |\beta_{nk}|^2 \, .
\end{equation}
We see that the vacuum is unstable if the coefficients $\beta_{nk}$ do not vanish. In the next section we will discuss formulas for the R\'enyi and entanglement entropies in the $out$ region in terms of the coefficients of the Bogoliubov transformation. Before that, we review the calculation of the Bogoliubov coefficients for the dynamical Casimir effect at parametric resonance performed in \cite{Dodonov:1996zz}.

\subsection{Dynamical Casimir effect at resonance}

Let $n$ be a multi-index $n=(n_1, \dots, n_d)$. The $in$-modes $\psi^n$ are solutions of the Klein-Gordon equation \eqref{eq:KG} that reduce to standing waves in the $in$-region:
\begin{equation}
    \psi^n = \left[ \prod_{i=1}^d \frac{1}{\sqrt{\pi n_i}} \sin \left(\frac {n_i \pi}{L_i} x^i \right) \right] e^{- i \omega_n t} \, , \quad \text{for }t\leq0 \, ,
\end{equation}
where
\begin{equation}
    \omega_n = \pi \sqrt{\sum_{i=1}^d \left( \frac{n_i}{L_i} \right)^2} \, .
\end{equation}

While the mirror is moving, the modes evolve in a nontrivial way. At each time $t$, we can expand them in a basis of instantaneous Fourier modes that satisfy the imposed boundary conditions,
\begin{align} \label{eq:psi-evol-inst-basis}
    \psi^n = &\sum_k  \left[ \prod_{i=1}^d \frac{1}{\sqrt{\pi n_i}} \sin \left(\frac {n_i \pi}{L_i} x^i \right) \right]  \nonumber \\
    & \qquad \times \sqrt{\frac{L_1}{L(t)}} \frac{1}{\sqrt{\pi n_1}} \sin \left(\frac {n_1 \pi}{L(t)} x^1 \right) Q^n_k(t) \, .
\end{align}
Substituting this expansion into the Klein-Gordon equation \eqref{eq:KG}, one finds equations of motion for the Fourier amplitudes $Q^n_k(t)$, which we will discuss in a moment. The Fourier amplitudes satisfy the initial conditions
\begin{equation}
    Q^n_k(0) = \delta_{kn} \, , \qquad \dot{Q}^n_k(0) =- i \omega_n \delta_{kn} \, .
\end{equation}

In the $out$-region, the evolution is again trivial. The positive energy $out$-modes are defined as
\begin{equation}
    \psi_k = \left[ \prod_{i=1}^d \frac{1}{\sqrt{\pi k_i}} \sin \left(\frac {k_i \pi}{L_i} x^i \right) \right] e^{- i \omega_k t} \, , \quad \text{for } t \geq T \, ,
\end{equation}
The $in$-mode $\psi^n$ evolves into a superposition of $out$-modes of positive and negative energy, as written in Eq.~\eqref{eq:psi-in-out}, with constant Bogoliubov coefficients $\alpha,\beta$ determined by:
\begin{equation}
    Q^n_k(t) = \alpha_{nk} e^{-i \omega_k t} + \beta_{nk} e^{i \omega_k t} \, , \quad \text{for } t\geq T \, .
\end{equation}

Up to this point, we considered a generic motion $L(t)$ of the boundary. Now let the moving boundary oscillate harmonically around its initial position.  When the frequency of the boundary oscillation is twice that of some normal mode, the system is said to be at parametric resonance. Explicit formulas for the Bogoliubov coefficients can then be obtained under an approximation of slow variation of the coefficients, as first shown in \cite{Dodonov:1996zz}, and we review the relevant results in the following subsections.

\subsubsection{(1+1) dimensions}
In one spatial dimension, all multi-indices reduce to integer numbers, and the expansion \eqref{eq:psi-evol-inst-basis} of the $in$-modes in the instantaneous basis becomes
\begin{equation} \label{eq:psi-evol-inst-basis-1d}
    \psi^n = \sum_k \sqrt{\frac{L_1}{L(t)}} \frac{1}{\sqrt{\pi n}} \sin \left(\frac {n \pi}{L(t)} x \right) Q^n_k(t) \, .
\end{equation}
The evolution equation for the Fourier amplitudes reads
\begin{equation} \label{eq:eom-1d}
\begin{split}
    \ddot{Q}^n_k + \Omega_k^2(t) Q^n_k &= 2 \lambda(t) \sum_j g_{kj} \dot{Q}^n_j + \dot{\lambda}(t) \sum_j g_{kj} Q^n_j \\
    &\qquad+ \lambda^2(t) \sum_{j,l} g_{jk} g_{jl} Q^n_l \, , 
    \end{split}
\end{equation}
where
\begin{equation} \label{eq:def-omegat-lambda}
    \Omega_k(t) =  \frac{\pi k}{L(t)} \, , \quad \lambda(t) = \frac{\dot{L}(t)}{L(t)} \, ,
\end{equation}
and the coefficients $g$ form an antisymmetric matrix with components
\begin{equation}
    g_{jk} = (-1)^{k+j} \frac{2jk}{k^2-j^2} \, , \quad \text{for } j \neq k \, .
\end{equation}

Consider the case where the first normal mode is resonant, that is, let the moving boundary oscillate harmonically around its initial position
\begin{equation} \label{eq:resonant-omega}
    L(t) = L_1 \left[ 1 + \epsilon \sin(2 \omega_1 t) \right] \, ,
\end{equation}
with twice the frequency of the first normal mode,
\begin{equation}
    \omega_1 = \pi/L_1 \, .
\end{equation}
The functions $\Omega_k(t)$ and $\lambda(t)$ can then be determined from Eqs.~\eqref{eq:resonant-omega} and \eqref{eq:def-omegat-lambda}. The oscillation amplitude is assumed to be small compared to the initial width $L_1$ of the boundary, $\epsilon \ll 1$. To first order in $\epsilon$, the parameters $\lambda$ and $\dot{\lambda}$ in the evolution equation oscillate harmonically with frequency $2 \omega_1$, and $\lambda^2 \sim \epsilon^2$ can be neglected.

One can look for solutions of the evolution equation \eqref{eq:eom-1d} of the form
\begin{equation} \label{eq:Q-slow}
	Q^n_k(t) = \alpha_{nk}(t) e^{- i \omega_k t} + \beta_{nk}(t) e^{i \omega_k t} \, ,
\end{equation}
where the Bogoliubov coefficients are now allowed to be time-dependent, reaching their asymptotic values $\alpha_{nk},\beta_{nk}$ at $t=T$. The time-dependent coefficients are assumed to vary slowly in time: after substituting \eqref{eq:Q-slow} into Eq.~\eqref{eq:eom-1d}, terms proportional to $\ddot{\alpha}$ and $\ddot{\beta}$ are neglected, and the Bogoliubov coefficients are considered approximately constant during one period of oscillation of any normal mode. Multiplying the resulting equation of motion by $e^{i \omega_k t}$ or $e^{-i \omega_k t}$ and averaging over intervals $T_k=2\pi/\omega_k$, one obtains a set of coupled differential equations for $\alpha$ and $\beta$, respectively. These are solved in \cite{Dodonov:1996zz}. Let us gather the relevant results for our purposes.

The complete elliptic integrals of the first and second kind are defined, respectively, as
\begin{align}
	K(\kappa) &= \int_0^{\pi/2} d\alpha \frac{1}{\sqrt{1-\kappa^2 \sin^2 \alpha}}\label{eq:elliptic K} \\ 
	E(\kappa) &=\int_0^{\pi/2} d\alpha \sqrt{1-\kappa^2 \sin^2 \alpha}\label{eq:elliptic E} \, .
\end{align}
It is convenient to introduce the new time variable
\begin{equation}
	\tau = \frac{1}{2} \epsilon \omega t \, ,
\end{equation}
and define the quantities
\begin{equation}\label{eq:kappas}
	\kappa = \sqrt{1-e^{-8\tau}} \, , \qquad \tilde{\kappa} = \sqrt{1-\kappa^2} = e^{-4\tau} \, .
\end{equation}
The lowest Bogoliubov coefficients are then given by
\begin{align} \label{eq:lowest-coefficients}
	\alpha_{11} &= \frac{2}{\pi} \frac{E(\kappa) + \tilde{\kappa} K(\kappa)}{1+\tilde{\kappa}} \, , \\
	\beta_{11} &= - \frac{2}{\pi} \frac{E(\kappa) - \tilde{\kappa} K(\kappa)}{1-\tilde{\kappa}} \, .
\end{align}
Coefficients with one larger odd index are obtained from the recurrence relations
\begin{align*}
	\sqrt{3} \, \alpha_{31} &= - \beta_{11} - \dot{\alpha}_{11} \, , \\
	\sqrt{3} \, \beta_{31} &= - \alpha_{11} - \dot{\beta}_{11} \, , \\
	\sqrt{n(n+2)} \, \alpha_{n+2,1} &= \sqrt{n(n-2)} \, \alpha_{n-2,1} - \dot{\alpha}_{n1} \, , 	\quad n \geq 3 \, , \\
	\sqrt{n(n+2)} \, \beta_{n+2,1} &= \sqrt{n(n-2)} \, \beta_{n-2,1} - \dot{\beta}_{n1} \, , 	\quad n \geq 3 \, , 
\end{align*}
where the dots represent derivatives with respect to $\tau$, and the relations
\begin{align*}
	\alpha_{1,2j+1} = (-1)^j (2j+1) \alpha_{2j+1,1} \, , \\
	\beta_{1,2j+1} = (-1)^j (2j+1) \beta_{2j+1,1} \, .
\end{align*}
All coefficients with some even index vanish. In order to compute the evolved state of the subsystem formed only by the resonant mode, it is enough to know the Bogoliubov coefficients with some index equal to $1$, which are all determined by the relations above.

In our analysis of entanglement production, the following formulas for infinite sums of products of Bogoliubov coefficients will play an important role:
\begin{align}
	& \sum_{n=1}^\infty \alpha_{n1} \dot{\alpha}_{n1} = \sum_{n=1}^\infty \beta_{n1} 	\dot{\beta}_{n1} = - \alpha_{11} \beta_{11} \, , 
	\label{eq:sum-alpha-dot-alpha}\\
	& \sum_{n=1}^\infty \left( \alpha_{n1} \dot{\beta}_{n1} + \dot{\alpha}_{n1} \beta_{n1} 	\right)=  - \left( \alpha_{11}^2 + \beta_{11}^2  \right) \, , 
	\label{eq:sum-alpha-dot-beta}\\
	& \sum_{k,n=1}^\infty \beta_{nk} \dot{\beta}_{nk} = -\sum_{n=1}^\infty \alpha_{n1} 	\beta_{n1} \, .
	\label{eq:double-sum-beta-dot-beta}
\end{align}

\subsubsection{(d+1) dimensions}\label{sec:d+1-dim}
In higher dimensions, the expansion of the $in$-modes in the instantaneous basis has the general form \eqref{eq:psi-evol-inst-basis}. Upon substitution of this expansion into the Klein-Gordon equation, one obtains equations of motion for the Fourier coefficients $Q^n_k$ similar to \eqref{eq:eom-1d}. The difference is that $j,k,n$ are now multi-indices, the frequencies of the normal modes are
\begin{equation} \label{eq:omega-d}
	\Omega_k(t) = \pi \sqrt{\sum_{i=1}^d \left( \frac{k_i}{L_i(t)} \right)^2} \, ,
\end{equation}
and the coefficients $g_{jk}$ have a more complicated form, which is not relevant for our purposes, except for the fact that they are still constant. 

A mode $r$ is set at resonance by letting the boundary oscillate with twice the time-independent frequency $\omega_r$ of the mode. One may take, for instance,
\begin{equation}
	L(t) = L_1 \left[1- \epsilon \cos(2 \omega_r t) \right]) \, 
\end{equation}
with
\begin{equation} \label{eq:omega-r}
	\omega_r = \pi \sqrt{\sum_{i=1}^d \left( \frac{k_i}{L_i} \right)^2} \, ,
\end{equation}
as done in \cite{Dodonov:1996zz}. Note that, in contrast to the one-dimensional case, the frequencies are not equidistant for $d>1$. To first order in $\epsilon$, the parameters $\lambda$ and $\dot{\lambda}$ in the evolution equation oscillate harmonically with frequency $2 \omega_r$, and $\lambda^2$ can be neglected, as before.

One can look again for solutions of the form \eqref{eq:Q-slow} under a slow-variation approximation. The resulting equations for the $Q$'s are quite different, however, as compared to the one-dimensional case. The reason for that is the following. When multiplied by $e^{\pm i \omega_k t}$, to first order in $\epsilon$, the right-hand-side of Eq.~\eqref{eq:eom-1d} becomes a sum of terms proportional to $\exp[i(\pm \omega_j \pm \omega_k \pm 2 \omega_r)]$. The integral over a period $T_k=2\pi/\omega_k$ of each such term vanishes unless the sum of frequencies in the exponential vanishes. But since they are not equally spaced, this never happens, and the equations for the distinct modes all decouple. One is left with an infinite set of equations for independent oscillators with time-dependent frequencies
\begin{equation}
	\Omega_k(t) \sim \omega_k \left[ 1+ 2 \gamma \cos(2 \omega_r t)\right] \, ,
\end{equation}
with
\begin{equation} \label{eq:gamma}
    \quad \gamma = \frac{\epsilon}{2} \frac{\pi^2 (k_1/L_1)^2}{\omega_k^2} \, .
\end{equation}

The calculation of the Bogoliubov coefficients can be done independently for each mode. Put $Q_k= Q^k_k$. Fourier amplitudes $G^n_k$ with $k\neq n$ vanish, since the modes are decoupled. We need to solve
\begin{equation} \label{eq:resonant-oscillator-eom}
    \ddot{Q}_k + \Omega_k^2(t) Q_k = 0 \, ,
\end{equation}
the equation of a time-dependent oscillator with a harmonically oscillating frequency. One can look again for solutions of the form \eqref{eq:Q-slow}, which we now write as
\begin{equation}
    Q_k(t) = \alpha_k(t) e^{- i \omega_k t} + \beta_k(t) e^{i \omega_k t} \, ,
\end{equation}
with $\alpha_k=\alpha_{kk}$ and $\beta_k=\beta_{kk}$. Under a slow variation approximation, terms proportional to $\ddot{\alpha}_k$ and $\ddot{\beta}_k$ are neglected. After averaging over the fast oscillations, one finds for the resonant mode
\begin{align} \label{eq:higher-d-alpha-beta-dot}
    \dot{\alpha}_r= -i \omega_r \gamma \beta_r \,,\qquad
    \dot{\beta}_r= i \omega_r \gamma \alpha_r \,,
\end{align}
while the evolution of the non-resonant modes is trivial in this approximation: $\alpha_k,\beta_k=\textrm{constant}$, for $k\neq r$. Initial conditions corresponding to the $in$-vacuum are given by
\begin{equation}
    \alpha_r(0)=1 \, , \quad \beta_r(0)= 0\, .
\end{equation}
Integrating Eq.~\eqref{eq:higher-d-alpha-beta-dot} with these initial conditions, we obtain the desired Bogoliubov coefficients
\begin{align} \label{eq:resonant-Bogoliubov-2d}
    \alpha_r(t) &= \cosh(\omega_r \gamma t ) \, , \nonumber \\
    \beta_r(t) &= i \sinh(\omega_r \gamma t) \, . 
\end{align}

The time-dependent symplectic transformation associated with such Bogoliubov coefficients, which describes the evolution of the resonant mode, is given by
\begin{equation} \label{eq:resonant-mode-evolution-d}
    M_r(t) = \begin{pmatrix}
    \cosh(\omega_r \gamma t ) & - i \sinh(\omega_r \gamma t) \\
    i \sinh(\omega_r \gamma t) & \cosh(\omega_r \gamma t ) 
    \end{pmatrix} \, ,
\end{equation}
and has the simple form
\begin{equation} \label{eq:generator-r-mode-2p1}
    M_r(t) = \exp(t K_r) \, , \quad K_r = \begin{pmatrix} 
    0 & -i \omega_r \gamma \\ i \omega_r \gamma & 0
    \end{pmatrix}
\end{equation}
with a time-independent symplectic generator $K_r$.

\section{Entanglement dynamics of Gaussian states}\label{sec:Ent.G. states}
In this section, we review the basic properties of Gaussian states. We focus on systems with a finite number of bosonic degrees of freedom, which we can later take to infinity to recover a bosonic field theory. In particular, we review compact expressions for the entropy and the von Neumann entanglement entropy in terms of the covariance matrix of the Gaussian state. Our conventions closely follow~\cite{Bianchi:2015fra,Bianchi:2017kgb,hackl2018aspects}, while other reviews include~\cite{eisert2003introduction,weedbrook_gaussian_2012}.

\subsection{Gaussian states}
Before considering the field theory case, we focus on a system with $N$ degrees of freedom. Here, the classical phase space $V\simeq\mathbb{R}^{2N}$ is equipped with a symplectic form $\Omega^{ab}$, \ie an antisymmetric and non-degenerate bilinear form. For the quantum theory, we choose the basis
\begin{align}
    \hat{\xi}^a\equiv(\hat{b}_1,\cdots,\hat{b}_N,\hat{b}_1^\dagger,\cdots,\hat{b}_N^\dagger)\label{eq:xi-basis}
\end{align}
of creation and annihilation operators. This basis also fixes the form of the symplectic form to get the correct commutation relations
\begin{align}
    [\hat{\xi}^a,\hat{\xi}^b]=\ii \Omega^{ab}\quad\text{with}\quad\Omega\equiv -\ii\begin{pmatrix}
    0 & \mathbb{1}\\
    -\mathbb{1} & 0
    \end{pmatrix}\,.\label{eq:commutator-Omega}
\end{align}
The basis is not Hermitian, which means that we can define the $2N \times 2N$ matrix $C$, such that
\begin{align}
    \hat{\xi}^{\dagger a}=C^a{}_b\hat{\xi}^b\quad\text{with}\quad C\equiv\begin{pmatrix}
    0 & \mathbb{1}\\
    \mathbb{1} & 0
    \end{pmatrix}\,.
\end{align}
Given a linear observable $\mathcal{O}=f_a\hat{\xi}^a$, we need to require $f_a^*C^a{}_b=f_b$ for $\mathcal{O}$ to be a Hermitian operator.

We can now introduce a special class of states $\ket{\psi}$ that are fully characterized by their displacement vector $z^a$ and their covariance matrix $G^{ab}$, defined as
\begin{align}
    z^a&=\bra{\psi}\hat{\xi}^a\ket{\psi}\,,\\
    G^{ab}&=\bra{\psi}\hat{\xi}^a\hat{\xi}^b+\hat{\xi}^b\hat{\xi}^a\ket{\psi}-2z^az^b\,.
\end{align}
We refer to a state $\ket{\psi}$ as Gaussian state and label it as $\ket{G,z}$ if its linear complex structure $J$ satisfies the condition
\begin{align}
    J^2=-\mathbb{1}\quad\text{with}\quad J^a{}_b=G^{ac}\Omega^{-1}_{cb}\,,
\end{align}
where we introduced the inverse symplectic form $\Omega^{-1}_{ab}$ with $\Omega^{ac}\Omega^{-1}_{cb}=\delta^a{}_b$. The linear map $J: V\to V$ is called \emph{linear complex structure} because it represents the imaginary unit on the classical phase space, \ie it squares to minus identity. We can use directly $J$ to encode the state $\ket{G,z}$ as solution to the equation
\begin{align}
    \frac{1}{2}(\delta^a{}_b-\ii J^a{}_b)(\hat{\xi}^b-z^b)\ket{G,z}=0\,.\label{eq:JdefiningState}
\end{align}
While complex structures have been used to describe vacua of quantum fields in curved spacetime~\cite{ashtekar1975quantum,wald1994quantum} for many years, only recently they became a useful tool in quantum information to parametrize both bosonic~\cite{Bianchi:2015fra,Bianchi:2017kgb,hackl2018aspects,hackl2019minimal} and fermionic Gaussian states~\cite{vidmar2017entanglement,vidmar2018volume,hackl2018circuit,hackl2019average} in a unified manner. As an example, we can choose $z^a=0$ and the covariance matrix to be given by
\begin{align}
    G_0\equiv\begin{pmatrix}
    0 & \mathbb{1}\\
    \mathbb{1} & 0
    \end{pmatrix}\quad\Rightarrow\quad J_0=G_0\Omega^{-1}\equiv \ii\begin{pmatrix}
    \mathbb{1} & 0\\
    0 & -\mathbb{1}
    \end{pmatrix}\,,
\end{align}
which gives rise to the operator-valued vector
\begin{align}
    \frac{1}{2}(\delta^a{}_b-\ii J^a{}_b)(\hat{\xi}^b-z^b)\equiv (\hat{b}_1,\cdots,\hat{b}_N,0,\cdots,0)
\end{align}
that annihilates the state $\ket{G_0,0}$, \ie we just confirmed that the Gaussian state $\ket{G_0,0}$ is nothing else than the vacuum associated to the annihilation operators $\hat{b}_i$. For other choices of $G^{ab}$ and $z^a$, we will find that~\eqref{eq:JdefiningState} gives rise to a Bogoliubov transformation of the form
\begin{align}
    \hat{a}_i=\sum_j (\alpha_{ji}\hat{b}_j+\beta^*_{ji}\hat{b}_j^\dagger )+z_i\,,
\end{align}
where $\hat{a}_i$ represent a new set of annihilation operators annihilating the state $\ket{G,z}$ under consideration.

A defining property of bosonic Gaussian states is the well-known Wick's theorem, \ie that higher order $n$-point correlation functions can be efficiently computed from $1$- and $2$-point correlation functions. To state Wick's theorem, it is useful to define the connected $n$-point correlation function of the state $\ket{G,z}$ as
\begin{align}
    C^{a_1\cdots a_n}_n=\bra{G,z}(\hat{\xi}^{a_1}-z^{a_1})\cdots(\hat{\xi}^{a_n}-z^{a_n})\ket{G,z}\,.
\end{align}
We can then use the connected $2$-point function given by
\begin{align}
    C_2^{ab}=\frac{1}{2}(G^{ab}+\ii\Omega^{ab})\,,
\end{align}
to state Wick's theorem as the equality
\begin{align}
    C_{2n+1}^{a_1\cdots a_{2n+1}}&=0\,,\\
    \begin{split}
    C_{2n}^{a_1\cdots a_{2n}}&=\sum(\text{all $2$-contractions of }C_2)\\
    &=C_2^{a_1a_2}\cdots C_2^{a_{2n-1}a_{2n}}+\cdots\,,
    \end{split}
\end{align}
\ie the connected $n$-point correlation function for odd $n$ vanishes, while for even $n$, it is given by a sum over products of $C_2^{ab}$.

\emph{For the rest of this paper, we will focus on Gaussian states $\ket{G,z}$ with $z^a=0$, \ie they are centered at the origin of phase space. We will write $\ket{G}=\ket{G,0}$ to simplify notation.}

\subsection{Quadratic Hamiltonians} \label{sec:quadratic-H}
The most general class of Hamiltonians that preserve the family of Gaussian states, \ie for which $e^{t\hat{H}}\ket{G,z}$ is again Gaussian, takes the form
\begin{align}
    \hat{H}=\frac{1}{2}h_{ab}\hat{\xi}^a\hat{\xi}^b+f_a\hat{\xi}^a\,.
\end{align}
If we only look at Gaussian states $\ket{G,z}$ with $z^a=0$, we need to choose $f_a=0$ in our Hamiltonian to ensure that $z^a$ remains equal to zero under time evolution. We will therefore restrict ourselves to the class of quadratic Hamiltonians
\begin{align}
    \hat{H}=\frac{1}{2}h_{ab}\hat{\xi}^a\hat{\xi}^b\,,
\end{align}
where we can derive a condition on $h_{ab}$ to ensure that $\hat{H}$ is Hermitian. For $\hat{\xi}^{\dagger a}=C^a{}_b\hat{\xi}^b$, we find the condition
\begin{align}
    h_{ab}=(C^\intercal)_a{}^c h_{dc} C^d{}_b\,.
\end{align}
Furthermore, we can require $h_{ab}$ to be symmetric, \ie $h_{ab}=h_{ba}$, because the anti-symmetric part will only contribute a constant to the energy. Defining $U(t)=e^{-\ii t\hat{H}}$, we can use the well-known Baker-Campbell-Hausdorff identity to compute
\begin{align}
    U^\dagger(t)\hat{\xi}^aU(t)=M(t)^a{}_b\hat{\xi}^b\,,\label{eq:xi-evolution}
\end{align}
where the symplectic transformation is given by
\begin{align}
    M(t)=e^{t K} \, , \quad\text{with } K^a{}_b=\Omega^{ac}h_{cb}\,.
\end{align}
For the dynamical Casimir effects, we will need to consider quadratic Hamiltonians with explicit time dependence
\begin{align}
    \hat{H}(t)=\frac{1}{2}h(t)_{ab}\hat{\xi}^a\hat{\xi}^b\,,
\end{align}
which leads to a time evolution determined by the time-ordered exponentials
\begin{align}
    U(t)&=\mathcal{T}e^{-\ii\int^t_0\,\hat{H}(t')\,dt'}\,,\\
    M(t)&=\mathcal{T}e^{\int^t_0\,K(t')\,dt'}\quad\text{with } K(t)^a{}_b=\Omega^{ac}h(t)_{cb}\,.\label{eq:M(t)-time-ordered}
\end{align}
In practice, we will not compute $M(t)$ as a time-ordered exponential, but rather by solving the underlying equations of motion explicitly. Using~\eqref{eq:xi-evolution}, which also continues to be valid for Hamiltonian with explicit time dependence, we can compute the time evolution of the covariance matrix to be given by
\begin{align}
    G(t)^{ab}=M(t)^a{}_c\, G^{cd}\, M^\intercal(t)_d{}^b\,.
\end{align}
For our class of Gaussian states, we therefore find
\begin{align}
    \ket{G_t}=U(t)\ket{G_0}=\ket{M(t)G_0M^\intercal(t)}\,.
\end{align}
A special class of time-dependent Hamiltonians are \emph{Floquet} Hamiltonians, \ie time-periodic Hamiltonians with period $T$ satisfying
\begin{align}
    \hat{H}(t+T)=\hat{H}(t)\,.
\end{align}
Such Hamiltonians can be described stroboscopically by only looking at the evolved states at $t=n\,\tau$ with $n\in\mathbb{Z}$. For this, it is sufficient to compute $U(\tau)$ which allows us to define a time-independent effective Hamiltonian $\hat{H}_{\text{eff}}=\frac{1}{\tau}\log U(\tau)$. For quadratic Hamiltonians $\hat{H}(t)$, we find $\hat{H}_{\text{eff}}=\frac{1}{2}h^{\text{eff}}_{ab}\hat{\xi}^a\hat{\xi}^b$ to be also quadratic with
\begin{align}
    h^{\text{eff}}_{ab}=\frac{1}{T}\Omega^{-1}_{ac} \big(\log M(T)\big)^c{}_b\,,
\end{align}
where $M(\tau)$ can be evaluated as the time-ordered exponential~\eqref{eq:M(t)-time-ordered} or by integrating the equations of motion.

\subsection{Entanglement measures}
Given a pure Gaussian state $\ket{G}$ and a choice of subsystem $A\subset V$ (inducing a tensor product decomposition $\mathcal{H}=\mathcal{H}_A\otimes\mathcal{H}_B$), we can quantify the amount of quantum correlations between $A$ and its complement $B$ using different entanglement measures. We are particularly interested in the von Neumann entropy $S_A(\ket{G})$ and the R\'enyi entropy $R_A(\ket{G})$ of the reduced state
\begin{align}
    \rho_A=\mathrm{Tr}_{B}\ket{G}\bra{G}\,,
\end{align}
where we trace out the degrees of freedom in $B$.

One can show that the mixed state $\rho_A$ is also Gaussian, \ie it continues to satisfy Wick's theorem and can be efficiently parametrized by the restriction of $G^{ab}$ to the subsystem $A$. Mathematically speaking, we have a decomposition of the classical phase space $V$ as a direct sum $V=A\oplus B$, as well as of its dual $V^*=A^*\oplus B^*$, such that the bilinear form $G: V^*\times V^*\to\mathbb{R}$ can be meaningfully restricted to
\begin{align}
    G_A: A^*\times A^*: (a_1,a_2)\mapsto G(a_1,a_2)\,.
\end{align}
In practice, we use the matrix representation of $G$ with respect to a basis $(\hat{\xi}^a_A,\hat{\xi}^a_B)$
\begin{align}
    \hat{\xi}^a_A&\equiv(\hat{b}_1^A,\cdots,\hat{b}_{N^A},\hat{b}_1^{A\dagger},\cdots,\hat{b}^{A\dagger}_{N_A})\,,\\
    \hat{\xi}^a_B&\equiv(\hat{b}_1^B,\cdots,\hat{b}_{N_B}^B,\hat{b}_1^{B\dagger},\cdots,\hat{b}^{B\dagger}_{N_B})\,,
\end{align}
with $\hat{\xi}^a=(\hat{\xi}^a_A,\hat{\xi}^a_B)$ and $N=N_A+N_B$. Note that this basis is complex, which implies that also $G$ correspond to the analytically continued bilinear form on the complexified phase space $V_{\mathbb{C}}$ and its dual $V^*_{\mathbb{C}}$.

The key result is that $\rho_A$ is fully characterized by $G_A$, which implies that any function of $\rho_A$ can be written purely in terms of $G_A$. The von Neumann entropy $S_A(\ket{G})$ and the R\'enyi entropy $R^{(n)}_A(\ket{G})$ of order $n$ are
\begin{align}
    S_A(\ket{G})&=-\mathrm{Tr}\rho_A\log{\rho_A}\,,\\
    R^{(n)}_A(\ket{G})&=\frac{1}{1-n}\mathrm{Tr}(\rho_A^n)\,,
\end{align}
and we define $R_A(\ket{G}):=R_A^{(2)}(\ket{G})$ as \emph{the} R\'enyi entropy, dropping the reference to order 2 in this case. To express these functions in terms of $G_A$, it is useful to define the restricted linear complex structure
\begin{align}
    (J_A)^a{}_b=-(G_A)^{ac}(\Omega_A)^{-1}_{cb}\,,
\end{align}
where $\Omega_A$ is the restriction of $\Omega$ to the subsystem $A$. An important difference between the full linear complex structure $J$ (describing a pure Gaussian state) and the restricted linear complex structure $J_A$ (describing in general a mixed Gaussian state) lies in the fact that $-J_A^2$ is not necessarily the identity---instead, we have the inequality
\begin{align}
    -J_A^2\geq\mathbb{1}\,,
\end{align}
\ie $J_A^2$ has negative eigenvalues whose magnitude is at least $1$, but possibly larger. The magnitude of these eigenvalues or, put differently, the failure of $J_A$ to square to minus identity provides a measure of entanglement between $A$ and $B$ in the state $\ket{G}$. The formulas for $S_A(\ket{G})$ and $R_A(\ket{G})$ are then given by~\cite{sorkin1983entropy,Bianchi:2015fra,Hackl:2017ndi}
\begin{align}
    R_A(\ket{G})&=\frac{1}{2} \log \det{J_A}=\frac{1}{2} \log \frac{\det{G_A}}{\det{\Omega_A}}\,,\label{eq:RA-formula}\\ 
    S_A(\ket{G})&=\textrm{Tr}\left(\frac{\mathbbm{1}+\text{i} J_A}{2}\log\left|\frac{\mathbbm{1}+\text{i} J_A}{2}\right|\right)\,, \label{eq:SA-formula}
\end{align}
where the absolute value $|\cdot|$ in the second equation is meant in terms of eigenvalues. Note that in the basis from~\eqref{eq:xi-basis} leads to the standard form of $\Omega$ (and $\Omega_A$) from~\eqref{eq:commutator-Omega}, which leads to $\det\Omega_A=(-1)^{N_A}$.

\subsection{Theorems on entanglement asymptotics}
\label{sec:asymptotic-entanglement}
For classically unstable quadratic systems and Gaussian initial states, there exist several theorems that describe the asymptotic behavior of the entanglement entropy and the R\'enyi entropy for large times. For this, let us introduce the concept of \emph{Lyapunov exponents} and \emph{unstable quadratic systems}.

\begin{definition}[Lyapunov exponents]
Given a quadratic Hamiltonian system $H(t)=\frac{1}{2}h(t)_{ab}\xi^a\xi^b$ with linear Hamiltonian flow
\begin{align}
    M(t)=\mathcal{T}\exp\int_0^tdt' K(t')
\end{align}
(where $K(t)^a{}_b=\Omega^{ac}h(t)_{cb}$) on the classical phase space $V$, we define the Lyapunov exponent $\lambda_{\ell}$ of a linear observable $\ell\in V^*$ as the limit
\begin{align}
    \lambda_{\ell}=\lim_{t\to\infty}\log\frac{\lVert M^\intercal(t)\ell\rVert}{\lVert \ell\rVert}\,,
\end{align}
where the norm $\lVert\cdot\rVert$ can be induced by any positive definite inner product on the dual phase space. Note that the action on the dual phase space is given by the transpose $M^\intercal(t)_a{}^b=M(t)^b{}_a$.
\end{definition}

\begin{definition}[Unstable quadratic system]
A classical quadratic Hamiltonian $H(t)=\frac{1}{2}h(t)_{ab}\xi^a\xi^b$ is called \emph{unstable} if there exists at least one linear observable $\ell\in V^*$ with positive Lyapunov exponent $\lambda_\ell>0$. There are at most $2N$ distinct Lyapunov exponents and we can always find a symplectic basis $\mathcal{D}_L=(\ell_1,\cdots,\ell_{2N})$ with associated Lyapunov exponents $\lambda_i:=\lambda_{\ell_i}$ satisfying
\begin{align}
    \lambda_1\geq\cdots\geq\lambda_{N}\geq 0\geq \lambda_{N+1}\geq\cdots\geq \lambda_{2N}
\end{align}
and $\lambda_{2N+1-i}=-\lambda_i$, \ie all Lyapunov exponents come in conjugate pairs.
\end{definition}
\begin{proof}
The proof can be found in definition 2 of appendix~A.2 in~\cite{Bianchi:2017kgb}.
\end{proof}

With these definitions at hand, we can now review the key theorems that we will use in subsequent sections to study dynamics of R\'enyi entropy (of order 2) and von Neumann entanglement entropy.

\begin{theorem}[R\'enyi entropy as phase space volume]\label{th:RAvolume}
The R\'enyi entropy (of order $2$) $R_A(\ket{G})=R_A^{(2)}(\ket{G})$ is given by
\begin{align}
    R_A(\ket{G})=\log\mathrm{Vol}_G\mathcal{V}_A\,,
\end{align}
where $\mathcal{V}_A\subset A^*\subset V^*$ represents an arbitrary parallelepiped of symplectic volume $1$.
\end{theorem}
\begin{proof}
The proof can be found in Sec.~6.2 of~\cite{Bianchi:2017kgb}.
\end{proof}

\begin{theorem}[Asymptotic entanglement production]\label{th:SAasymptotics}
We consider an \emph{unstable} and \emph{regular}\footnote{The word ``regular'' is used to exclude some pathological cases discussed in Appendix A.3 of~\cite{Bianchi:2017kgb}.} Hamiltonian system with $N$ degrees of freedom and quadratic Hamiltonian $\hat{H}(t)=\frac{1}{2}h(t)_{ab}\hat{\xi}^a\hat{\xi}^b$, whose Lyapunov exponents are
\begin{align}
    \lambda_1\geq\cdots\geq\lambda_{N}\geq 0\geq \lambda_{N+1}\geq\cdots\geq \lambda_{2N}\,.
\end{align}
An initial Gaussian state $\ket{G_0}$ will evolve into $\ket{G_t}$, where the time-dependent covariance matrix is given by
\begin{align}
    G_t=M(t)G_0M^{\intercal}(t)\,.
\end{align}
A system decomposition $V=A\oplus B$ of the classical phasespace $V$ (with $\dim A=2N_A$ and $\dim B=2N_B$) induces a tensor product decomposition $\mathcal{H}=\mathcal{H}_A\otimes\mathcal{H}_B$ that we can use to compute entanglement entropy $S_A(\ket{G_t})$. The asymptotics of the R\'enyi entropy $R_A(\ket{G_t})$ and the entanglement entropy $S_A(\ket{G_t})$ is given by
\begin{align}
    R_A(\ket{G_t})\sim S_A(\ket{G_t})\sim \Lambda_A\,t \, , \quad\text{with}\quad \Lambda_A=\sum^{2N_A}_{i=1}\lambda_i\,,
\end{align}
where we assume the subsystem to be generic, \ie above statement applies to all subsystems except a set of measure zero.
\end{theorem}
\begin{proof}
A detailed proof can be found in Sec.~2.1 and~2.2 of~\cite{Bianchi:2017kgb}.
\end{proof}

Next, we will compare the rate of entanglement entropy production in the dynamical Casimir effect with these predictions.

\section{Entanglement production at parametric resonance}\label{sec:Ent.Prod}
In this section, we compute the time-dependent R\'enyi and entanglement entropies at resonance. In $(1+1)$-D, the normal modes are strongly coupled by the evolution and become highly entangled. We compute the entropies explicitly for the subsystem associated with the resonant mode. In higher dimensions, the modes decouple under the assumed approximations, and the normal modes are not entangled. In this case, we consider a generic subsystem that intersects nontrivially the resonant mode.

\subsection{(1+1) dimensions}
Let the scalar field be initially in its vacuum state $\ket{0,in}$. We work in the Heisenberg representation, so that the state of the system is constant and the time evolution is encoded in the observables. Annihilation and creation operators $\hat{a}_k(t),\hat{a}_k^\dagger(t)$ for a mode $k$ at each time $t$ are defined in terms of the time-dependent Bogoliubov coefficients as in Eq.~\eqref{eq:B-transf-complex} :
\begin{equation} \label{eq:t-Bogoliubov}
    \hat{a}_k(t) = \sum_n \left[ \alpha_{nk}(t) \, \hat{b}_n + \beta_{nk}^*(t) \, \hat{b}_n^\dagger  \right] \, .
\end{equation}
Such operators are associated with the instantaneous basis at $t$. It is convenient to use the variable $\tau=\epsilon \omega t/2$ to compute the evolution of the entropies.

The covariance matrix at time $\tau$ is given by
\begin{align} \label{eq:covariance-matrix}
    G^{ab}(\tau)=\bra{0,in} ( \hat{\xi}^a\hat{\xi}^b+\hat{\xi}^b\hat{\xi}^a ) \ket{0,in}\,,
\end{align}
where $\hat{\xi}^a = (\hat{a}_k(\tau),\hat{a}_k^\dagger(\tau))$ in the complex representation. The restriction of the covariance matrix to a subsystem $A$ defines the reduced covariance matrix $G_A$ associated with it. The R\'enyi entropy $R_A(\ket{G})$ can be computed from $G_A$ as
\begin{equation}
	R_A = \frac{1}{2} \log (-\det G_A) \,,
\end{equation}
where we used Eq.~\eqref{eq:RA-formula} and the fact that $\det(\Omega_A)=-1$ for a single degree of freedom ($N_A=1$). Similarly, the entanglement entropy is given by
\begin{equation}
	S_A(\ket{G})=\textrm{Tr}\left(\frac{\mathbbm{1}+\text{i} J_A}{2}\log\left|\frac{\mathbbm{1}+\text{i} J_A}{2}\right|\right)
\end{equation}
With $J_A=-G_A\,\omega^{-1}_A$ as explained in Eq.~\eqref{eq:SA-formula}. We will now compute these quantities.

Expressing $\hat{\xi}^a$ in terms of creation and annihilation $in$-operators $\hat{b}_m^\dagger$ and $\hat{b}_m$ using Eq.~\eqref{eq:t-Bogoliubov}, and representing the covariance matrix as
\begin{align}
	G_A(\tau)  = \begin{pmatrix}
	G_A^{11} & G_A^{12}\\
	G_A^{21} & G_A^{22}
	\end{pmatrix}\,,
\end{align}
we find for the matrix components
\begin{align} \label{eq:G-alpha-beta}
	G_A^{11} &= \sum_n 2\alpha_{n1}(\tau)\beta_{n1}^*(\tau) \nonumber \\
	G_A^{12} = G_A^{21} &= \sum_n \alpha_{n1}(\tau)\alpha_{n1}^*(\tau)+\beta_{n1}(\tau)\beta^*_{n1}(\tau) \nonumber \\
	G_A^{22} &= \sum_n 2\beta_{n1}(\tau)\alpha^*_{n1}(\tau) \, .
\end{align}
From now on, we will omit the arguments of the time-dependent Bogoliubov coefficients for conciseness.

Since $\alpha_{nm}$ and $\beta_{nm}$ are real for any $m$ and $n$, we have
\begin{align}
    G_A^{11}=G_A^{22} &= \sum_n 2\alpha_{n1}\beta_{n1} \\
    G_A^{12}=G_A^{21} &= \sum_n \left(\alpha_{n1}^2+\beta_{n1}^2 \right) \,.
\end{align}
Differentiating with respect to $\tau$, these sums can be evaluated using Eqs.~\eqref{eq:sum-alpha-dot-alpha}--\eqref{eq:double-sum-beta-dot-beta} and then integrated analytically.

Let us start with the term $G_A^{11}$. Taking its derivative with respect to $\tau$ yields
\begin{align}
    \frac{dG_A^{11}}{d\tau}&=\sum_n 2 \left( \Dot{\alpha}_{n1} \beta_{n1} + \alpha_{n1} \Dot{\beta}_{n1} \right)\,.
\end{align}
Then, according to Eq.~\eqref{eq:sum-alpha-dot-beta}, we have
\begin{align}\label{eq:g11-edo}
     \frac{dG_A^{11}}{d\tau}&= -2 \left( \alpha_{11}^2 + \beta_{11}^2  \right)\,.
\end{align}
Similarly, taking the derivative of the term $G^{12}_A$, we find:
\begin{align}
     \frac{dG_A^{12}}{d\tau}&=\sum_n 2 \left( \alpha_{n1} \Dot{\alpha}_{n1} + \beta_{n1} \Dot{\beta}_{n1} \right)\,.
\end{align}
From Eq.~\eqref{eq:sum-alpha-dot-alpha}, we conclude that
\begin{align}\label{eq:g12-edo}
    \frac{dG_A^{12}}{d\tau}&=-4 \alpha_{11} \beta_{11} \, .
\end{align}
Since $G^{11}_A=G^{22}_A$ and $G^{12}_A=G^{21}_A$, the only remaining step is to integrate Eqs.~\eqref{eq:g11-edo} and \eqref{eq:g12-edo}.

The explicit form of the Bogoliubov coefficients $\alpha_{11}$ and $\beta_{11}$ is presented in Eq.~\eqref{eq:lowest-coefficients}. Inserting these formulas in Eq.~\eqref{eq:g11-edo} we find
\begin{equation}\label{eq:g11-edo-2}
    \frac{dG_A^{11}}{d\tau}=-\frac{16\left\{[E(\kappa)-\Tilde{\kappa}^2K(\kappa)]^2+\Tilde{\kappa}^2[E(\kappa)-K(\kappa)]^2\right\}}{\kappa^4\pi^2}\,,
\end{equation}
where the complete elliptic integrals of the first and second kind, $K(\kappa)$ and $E(\kappa)$, are defined in Eqs.~\eqref{eq:elliptic E} and \eqref{eq:elliptic K}, and $\kappa$ and $\Tilde{\kappa}$ are given by Eq.~\eqref{eq:kappas}.

From the definitions of $\kappa$ and $\Tilde{\kappa}$, we can easily check that
\begin{equation}
    d\tau=\frac{\kappa d\kappa}{4\Tilde{\kappa}^2} \, .
\end{equation}
Substituting this relation in \eqref{eq:g11-edo-2}, we can integrate it analytically (see \cite{Dodonov:1996zz}), obtaining:
\begin{equation} \label{eq:g11-int-cte}
    G_A^{11}=-\frac{4}{\pi^2\kappa^2}\left[E(\kappa)-\Tilde{\kappa}^2K(\kappa)\right]\left[K(\kappa)-E(\kappa)\right]+C_1,
\end{equation}
where $C_1$ is a constant of integration to be determined from the initial conditions.

In a similar way, substituting \eqref{eq:lowest-coefficients}
in Eq.~\eqref{eq:g12-edo}, we find
\begin{equation}
    \frac{dG_A^{12}}{d\tau}=\frac{16}{\pi^2 \kappa^2} [E^2(\kappa)-\Tilde{\kappa}^2K^2(\kappa)].
\end{equation}
This equation can again be integrated \cite{Dodonov:1996zz}, leading to
\begin{equation} \label{eq:g12-int-cte}
    G_A^{12}=\frac{4}{\pi^2}E(\kappa)K(\kappa)+C_2\,,
\end{equation}
where $C_2$ is a constant to be determined from the initial conditions.

To fix the integration constants, we need to evaluate the expressions in Eqs.~\eqref{eq:g11-int-cte} and \eqref{eq:g11-int-cte} at $\tau=0$ and compare them with the components of the initial covariance matrix $G_A(0)$. But at $\tau=0$, the Bogoliubov transformation is trivial, $\alpha_{mn}(0)=\delta_{mn},\beta_{mn}(0)=0$. Hence, the initial covariance matrix is given by
\begin{equation}
    G_A(0)=\left(\begin{array}{cc}
    0     & 1 \\
    1     &0 
    \end{array}\right).
\end{equation}
Now, in the small-time limit $\tau\ll 1$, the elliptic integrals can be expanded as~\cite{gradshteyn2014table}
\begin{align}
    K(\kappa)&=\frac{\pi}{2}\left\{1+\frac{1}{4}\kappa^2+\frac{9}{64}\kappa^4 + \cdots \right\}\,,\\
    E(\kappa)&=\frac{\pi}{2}\left\{1-\frac{1}{4}\kappa^2-\frac{3}{64}\kappa^4- \cdots \right\}\,.
\end{align}
In addition, in this limit we see that $\kappa\to 0$ and $\Tilde{\kappa}\to 1$. Keeping only terms up to first order in $\tau$, we find the matrix elements
\begin{align}
    G_A^{11} &= G_A^{22}\sim-2\tau+C_1\,, \\
    G_A^{12} &= G_A^{21}\sim 1+C_2\,.
\end{align}
Hence, the integration constants are $C_1=C_2=0$, and the matrix $G_A(\tau)$ assumes the form
\begin{widetext}
\begin{equation}
    G_A(\tau)=\left(\begin{array}{cc}
     -\frac{4}{\pi^2\kappa^2}\left[E(\kappa)-\Tilde{\kappa}^2K(\kappa)\right]\left[K(\kappa)-E(\kappa)\right]    &\frac{4}{\pi^2}E(\kappa)K(\kappa)  \\
      \frac{4}{\pi^2}E(\kappa)K(\kappa)   & -\frac{4}{\pi^2\kappa^2}\left[E(\kappa)-\Tilde{\kappa}^2K(\kappa)\right]\left[K(\kappa)-E(\kappa)\right]
    \end{array}\right) \, .
\end{equation}
\end{widetext}

Furthermore, we can analyze the long-time behavior of the R\'enyi entropy. According to \cite{gradshteyn2014table}, when $\tau\gg 1$, the leading terms of the asymptotic expansion of the elliptic integrals are
\begin{align}
    K(\kappa)&\sim\log\frac{4}{\Tilde{\kappa}}+\frac{1}{4}\left(\log\frac{4}{\Tilde{\kappa}}-1\right)\Tilde{\kappa}^2+ \cdots\\
    E(\kappa)&\sim1+\frac{1}{2}\left(\log\frac{4}{\Tilde{\kappa}}-\frac{1}{2}\right)\Tilde{\kappa}^2+ \cdots
\end{align}
Therefore, the covariance matrix becomes
\begin{equation}
    G_A(\tau)=\frac{4}{\pi^2}\left(\begin{array}{cc}
      1-\zeta   &\zeta  \\
    \zeta     &1-\zeta 
    \end{array}\right),
\end{equation}
where $\zeta\equiv\log 4 + 4\tau$. We find that, in the asymptotic limit, the R\'enyi entropy is
\begin{align}
    R_A(\tau)&\sim\frac{1}{2}\log\left(\frac{16(8\tau+\log{16}-1)}{\pi^4}\right) \label{eq:RAasymp1}\\
    & \sim\frac{1}{2}\log\frac{128}{\pi^4}+\frac{1}{2}\log \tau\,.\label{eq:RAasymp2}
\end{align}
The R\'enyi entropy converges fast to its asymptotic behavior, as shown in Fig.~\ref{fig1}. In particular, we see in Fig.~\ref{fig2} that the first asymptotic expansion from Eq.~\eqref{eq:RAasymp1} approaches the exact solution exponentially fast, while the leading behavior from Eq.~\eqref{eq:RAasymp2} keeps a finite offset.

\begin{figure}[b!]
\centering
  \begin{tikzpicture}
  \draw (0,0) node[inner sep=0pt]{\includegraphics[width=\linewidth]{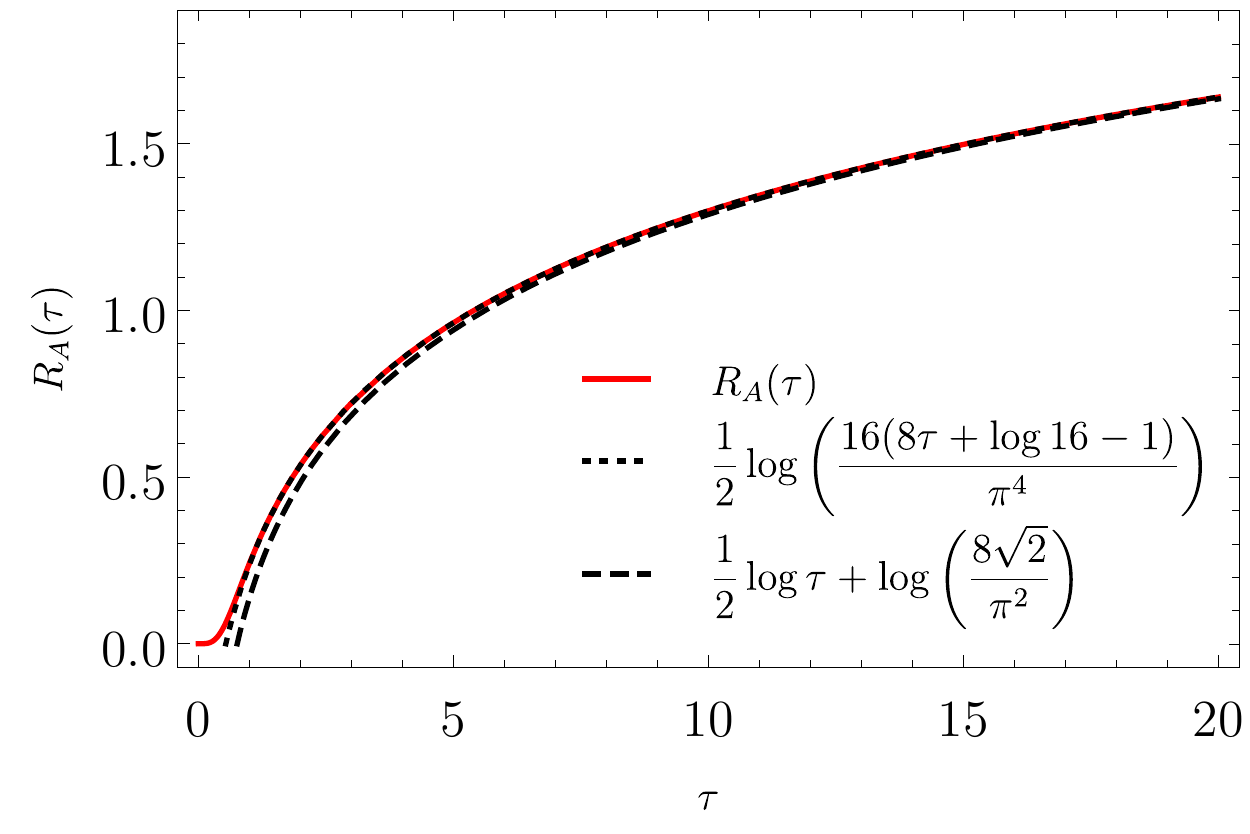}};
  \end{tikzpicture}
  \caption{Comparison of the exact time-dependent R\'enyi entropy $R_A(\tau)$ of a resonant mode in $(1+1)$ dimensions with its asymptotic expansions given in Eqs.~\eqref{eq:RAasymp1} and~\eqref{eq:RAasymp2}.}
  \label{fig1}
\end{figure}

\begin{figure}[t!]
\centering
  \begin{tikzpicture}
  \draw (0,0) node[inner sep=0pt]{\includegraphics[width=\linewidth]{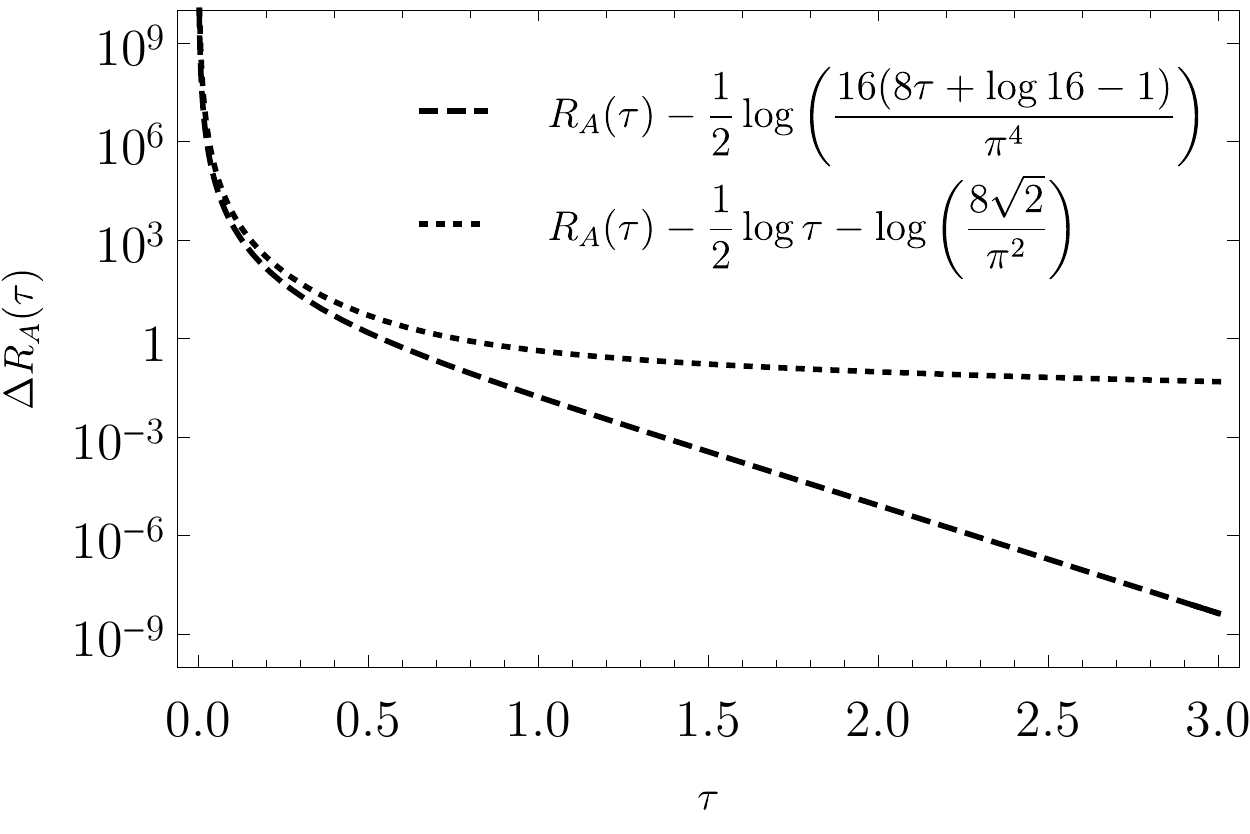}};
  \end{tikzpicture}
  \caption{Difference between the exact R\'enyi entropy $R_A(\tau)$ of a resonant mode in $(1+1)$ dimensions and the asymptotic expansions given in Eqs.~\eqref{eq:RAasymp1} and~\eqref{eq:RAasymp2}.}
  \label{fig2}
\end{figure}

For a single degree of freedom, the relation between the R\'enyi entropy $R_A$ and the entanglement entropy $S_A$ is given by
\begin{align}\label{EntRenyi}
    S_A(\tau)=s\left(e^{R_A(\tau)}\right)\, ,
\end{align}
with
\begin{equation}
    s(x)=\left(\frac{x+1}{2}\right)\log\left(\frac{x+1}{2}\right)-\left(\frac{x-1}{2}\right)\log\left(\frac{x-1}{2}\right) \, .
\end{equation}
This function behaves as
\begin{equation}
    s(x) \sim\log x + (1- \log 2)\quad\text{as}\quad x\to\infty\,.
\end{equation}
Hence, for large $x$, we have
\begin{equation} \label{eq:S-R-asymptotic}
    S_A(\tau) \sim R_A(\tau) + (1-\log 2)\quad\text{as}\quad x\tau\to\infty \, .
\end{equation}
Combining Eqs.~\eqref{eq:RAasymp2} and \eqref{eq:S-R-asymptotic}, we obtain the asymptotic behavior of the entanglement entropy given by
\begin{equation} \label{eq:SA-asymp}
    S_A(\tau) \sim 1+\frac{1}{2}\log\frac{32}{\pi^4}+\frac{1}{2}\log \tau \,,
\end{equation}
where we find the constant to be $1+\frac{1}{2}\log{32/\pi^2}\approx 0.44$.

\begin{figure}[b]
\centering
  \begin{tikzpicture}
  \draw (0,0) node[inner sep=0pt]{\includegraphics[width=\linewidth]{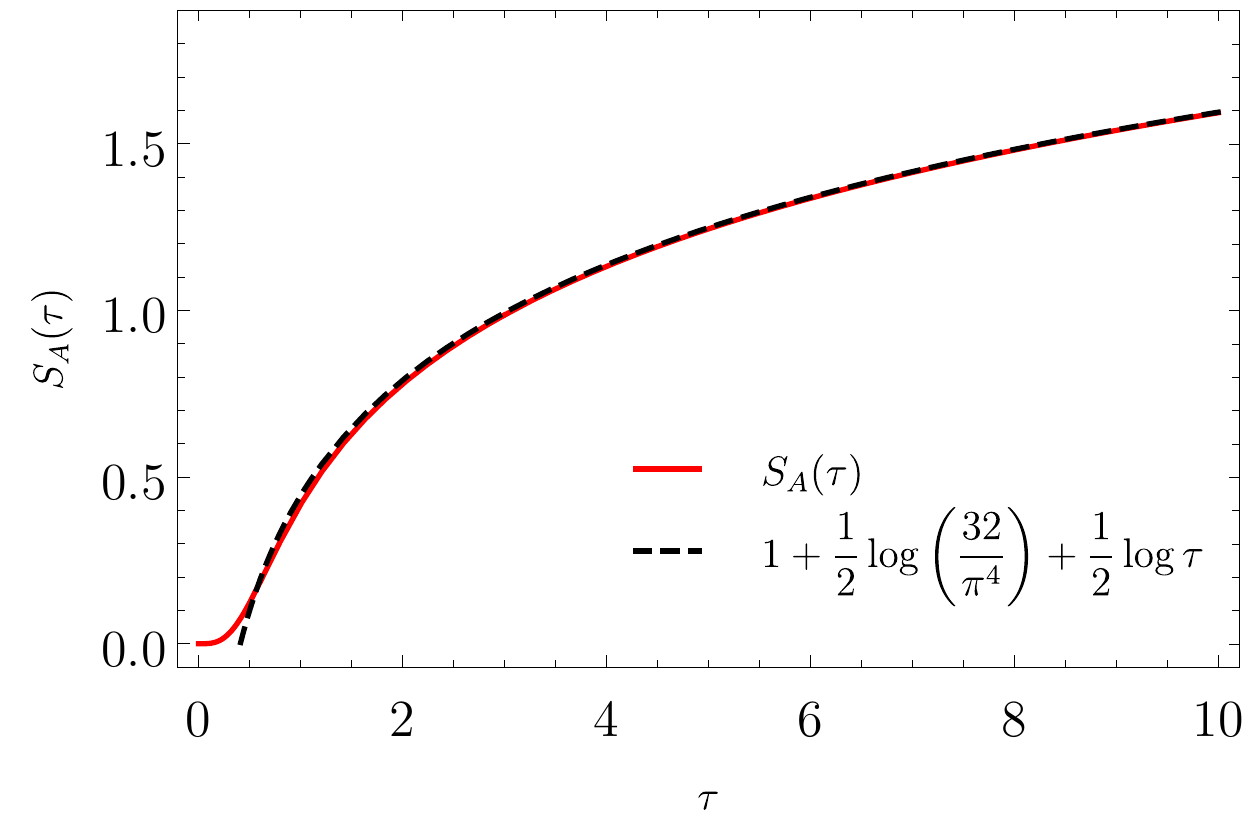}};
  \end{tikzpicture}
  \caption{Comparison of the time-dependent entanglement entropy $S_A(\tau)$ of a resonant mode in $(1+1)$ dimensions with its asymptotic expansion given in Eq.~\eqref{eq:SA-asymp}.}
  \label{fig3}
\end{figure}

In short, we found that in one spatial dimension, the asymptotic growth of the R\'enyi and entanglement entropies is logarithmic. As discussed in Section \ref{sec:quadratic-H}, to any system with a periodic quadratic Hamiltonian there is an associated effective time-independent Hamiltonian that describes its stroboscopic evolution. Moreover, if such an effective Hamiltonian displays a Lyapunov instability, then the asymptotic growth of the entropy is linear, according to the general results summarized in Section \ref{sec:asymptotic-entanglement}. Therefore, we conclude that for $d=1$ the DCE displays no Lyapunov instability, despite the presence of parametric resonance. As stressed in \cite{Dodonov:1996zz}, the case of $d=1$ is special in the sense that, after the averaging of the fast oscillations, the field modes remain strongly coupled. The intermode interactions compete with the effect of parametric resonance, keeping the production of particles in the resonant mode linear, in opposition to what happens in higher dimensions, where it grows  exponentially. Our result show that this suppression of particle creation is reflected in the absence of a Lyapunov instability for the stroboscopic evolution. 

The possible asymptotic behaviors of the entanglement entropy for systems with time-independent quadratic Hamiltonians and a finite number of degrees of freedom was fully analyzed in \cite{Hackl:2017ndi}. In the presence of instabilities, the asymptotic growth is linear, in agreement with the previous results obtained in \cite{Bianchi:2017kgb}, discussed in Section \ref{sec:asymptotic-entanglement}. In the absence of instabilities, the entanglement entropy either grows logarithmically, for so-called metastable systems, or is bounded. Hence, the DCE in $d=1$ behaves qualitatively as a metastable system, but there is an important difference. In \cite{Hackl:2017ndi}, it was proved that the prefactor multiplying the logarithmic function must be an integer, \ie $S_A(t) \sim n \log t$, $n \in \mathbb{N}$. In contrast, we found $S_A(t) \sim 1/2 \log t$, according to Eq.~\eqref{eq:SA-asymp}. This provides evidence that the extension of the results of \cite{Hackl:2017ndi} to systems with an infinite number of degrees of freedom is nontrivial, with new classes of asymptotic behavior arising in this limit. In fact, in preliminary numerical investigations, we found that, by truncating the number of degrees of freedom in the one dimensional DCE, the entanglement entropy becomes bounded, with an initial regime of logarithmic growth. As the size of the truncated system is increased, the regime of logarithmic growth becomes longer, consistently with the fact that for the infinite system such a regime has an infinite duration. This mechanism illustrates how a new class of logarithmic asymptotic growth can arise in the nontrivial limit of infinite number of degrees of freedom and may be investigated elsewhere.

\subsection{(2+1) dimensions} \label{sec:2+1-entropy}

We now proceed to the case of $2$ dimensions. The frequencies of the normal modes are determined by Eq.~\eqref{eq:omega-d}, where $k$ is now a set of two integers. We let the moving mirror oscillate with twice the frequency $\omega_r=\omega_{11}$ of the first mode $r=(1,1)$,
\begin{equation}
    L(t)=L_1[1-\epsilon\cos(2\omega_r t)],
\end{equation}
so that the resonant mode is $r$. The length of the cavity along the $x^2$-direction is kept fixed. The resonant frequency is
\begin{equation}
    \omega_r = \pi \sqrt{\frac{1}{L_1^2}+ \frac{1}{L_2^2}} \, .
\end{equation}

The R\'{e}nyi entropy can be computed from Eq.~\eqref{eq:RA-formula} as before. The main difference is that, as discussed in section \ref{sec:d+1-dim}, the evolution of the system is now much simpler. Under the slow-variation approximation, all normal modes decouple, and the evolution is nontrivial only for the resonant mode. The problem is thus reduced to that of a one dimensional parametric oscillator at resonance. The nontrivial Bogoliubov coefficients are given by Eq.~\eqref{eq:resonant-Bogoliubov-2d}
\begin{align}\label{Bogcoeff}
    \alpha_{11}(t)&=\cosh(\omega_r \gamma t),\qquad \nonumber \\ \beta_{11}(t)&=-i\sinh(\omega_r \gamma t)\,,
\end{align}
where $\gamma$, defined in \eqref{eq:gamma}, now reads
\begin{equation} \label{eq:gamma-2p1-d}
    \gamma= \frac{\epsilon \pi^2}{2 L_1^2 \omega_r^2 } = \frac{\epsilon}{2} \frac{L_2^2}{L_1^2 + L_2^2}  \, .
\end{equation}

We start again from the covariance matrix in order to compute the R\'{e}nyi entropy. The sums in the formulas \eqref{eq:G-alpha-beta} for the matrix components for each mode are now easily computed, since only the first term of each sum is nonzero. For the resonant mode $r=(1,1)$, we obtain
\begin{equation}\label{eq:2+1-Cov2x2}
    G_r(t)=\left(\begin{array}{cc}
    G_r^{11} &G_r^{12}  \\
    G_r^{21}    & G_r^{22}
    \end{array}\right) \, ,
\end{equation}
with
\begin{align}
    G_r^{11} &= - G_1^{11} = 2i\cosh(\omega_r\gamma t) \, , \\
    G_r^{12} &= G_1^{21} = \cosh^2(\omega_r \gamma t)+\sinh^2(\omega_r \gamma t) \, .
\end{align}
For the remaining modes, the covariance matrix keeps its initial form throughout the evolution,
\begin{equation}
    G_k(t)=\left(\begin{array}{cc}
    0 & 1  \\
    1 & 0
    \end{array}\right) \, , \quad k \neq (1,1) \, .
\end{equation}

In order to compute the entanglement entropy via Eq.~\eqref{eq:SA-formula}, we need to project the covariance matrix to a subspace of interest. Here, a special care must be taken so that a nontrivial evolution is obtained. Since distinct modes evolve independently, if we choose the subsystem to be a single mode, then the entropy will be zero. In order for the entropy to be nonzero, we should consider subspaces that intersect both the resonant mode and its complement.

 As a simple example, we study a mixture of two modes. In order to define the subsystem $A$ of interest, let us first introduce an extended subsystem  $E$ including the resonant mode and some other (arbitrary) mode $s$. The only mode with a nontrivial evolution is the resonant mode; for all others, the covariance matrix keeps its initial form. As a result, regardless of which mode $s$ one chooses to build the subsystem $E$, the projection of the covariance matrix to $E$ will always have the form
\begin{widetext}
\begin{equation}
    G_{E}(t)=\left(\begin{array}{cccc}
     2i\cosh(\omega_r\gamma t)    &0  &\cosh^2(\omega_r \gamma t)+\sinh^2(\omega_r \gamma t) &0\\
      0   &0  & 0&1\\
     \cosh^2(\omega_r \gamma t)+\sinh^2(\omega_r \gamma t)    &0  &-2i\cosh(\omega_r\gamma t) &0\\
       0  &1  &0 &0
    \end{array}\right).
\end{equation}
\end{widetext}

We now define a new set of annihilation operators
\begin{align}
    \hat{\Tilde{a}}_1=\frac{\hat{a}_1+\hat{a}_s}{\sqrt{2}},\qquad \hat{\Tilde{a}}_s=\frac{\hat{a}_1-\hat{a}_s}{\sqrt{2}}
\end{align}
and their conjugates $\hat{\tilde{a}}^\dagger_1$ and $\hat{\tilde{a}}^\dagger_s$. This symplectic transformation defines a new basis on the subsystem $E$. The transformation matrix $B$ that connects the new and the old basis is
\begin{equation}
    B=\frac{1}{\sqrt{2}}\left(\begin{array}{cccc}
        1 &1&0& 0 \\
        1 &-1&0&0\\
        0&0&1&1\\
        0&0&1&-1
    \end{array}\right).
\end{equation}
In order to mix the modes and find a nonzero entropy, we apply this transformation to $G_{E}$
\begin{equation}
    \Tilde{G}_{E}(t)=BG_{E}B^{-1},
\end{equation}
and then restrict to the subspace spanned by the first transformed mode, which gives the $2\times2$ restricted covariance matrix in the transformed basis, which we write as
\begin{equation}
    \Tilde{G}_A(t)=\left(\begin{array}{cc}
     \frac{1}{2}i\sinh(2\omega_r\gamma t)    &\cosh^2(\omega_r\gamma t)  \\
    \cosh^2(\omega_r\gamma t)     & -\frac{1}{2}i\sinh(2\omega_r\gamma t)
    \end{array}\right).
\end{equation}
It is now simple to compute the R\'{e}nyi entropy as
\begin{align}\label{Renyi2+1}
    R_{A}(t)&=\frac{1}{2}\log\left(-\det(\Tilde{G}_{A})\right) \nonumber \\
    &=\frac{1}{2}\log\left(\cosh^2(\omega_r \gamma t)\right).
\end{align}
In the limit of $t\gg 1$, the R\'{e}nyi entropy becomes
\begin{equation}\label{Renyi3d}
    R_A(t)\sim - \log 2 + \omega_r \gamma t \, .
\end{equation}

From Eq.~\eqref{eq:S-R-asymptotic}, the entanglement entropy has the same asymptotic behavior, differing only by an offset
\begin{equation} \label{eq:asympotic-ee-2d}
    S_A(t) \sim 1 - 2 \log 2 +\omega_r \gamma t \, .	
\end{equation}
Hence, entropy is produced at a constant rate for large times, with a rate determined by the frequency $\omega_r$ of the resonant mode, the amplitude of oscillation $\epsilon$ of the mirror, and the dimensions $L_1, L_2$ of the cavity, via the parameter $\gamma$. We see in Fig.~\ref{fig4} that, for large $t$, Eq.~\eqref{eq:asympotic-ee-2d} gives the same asymptotic behavior of the exact entanglement entropy computed from Eqs.~\eqref{EntRenyi} and \eqref{Renyi2+1}. 

This result can be interpreted as signaling the presence of an instability in the system, associated with the fast amplification of excitations in the resonant mode. As discussed in \cite{Bianchi:2017kgb}, in any system displaying Lyapunov instabilities, the asymptotic production of entropy for a generic subsystem takes place at a constant rate $\Lambda_A$, which is determined by the Lyapunov exponents of the system. We will now show that $\Lambda_A = \omega_r \gamma$ in the present case, in agreement with Eq.~\eqref{eq:asympotic-ee-2d}.

The asymptotic rate of entanglement entropy production $\Lambda_A$ for a generic subsystem with $N_A$ degrees of freedom is given by the sum of the largest $2N_A$ Lyapunov exponents \cite{Bianchi:2017kgb}. Since we are dealing with a single degree of freedom, $N_A=1$ here, and $\Lambda_A$ is the sum of the two largest Lyapunov exponents. We will later give a precise characterization of such generic subsystems.

The Lyapunov exponents can be computed from the Hamiltonian flow $M(t)$ of the system. In the special case where $M(t)=\exp(K t)$, they are simply the eigenvalues of the time-independent symplectic generator $K$. In our case, the generator decomposes into a series of $2\times 2$ blocks $K_k$ associated with the independent modes $k$, as discussed in Section \ref{sec:d+1-dim}. The evolution of the nonresonant modes is trivial, $K_k = 0$, for $k \neq r$, and the associated Lyapunov exponents vanish. The symplectic generator of the resonant mode is given in Eq.~\eqref{eq:generator-r-mode-2p1}. The Lyapunov exponents are the eigenvalues
\begin{equation} \label{eq:lyapunov-exponents}
    \lambda^{(r)}_1 = \omega_r \gamma \, , \qquad \lambda^{(r)}_2 = - \omega_r \gamma \, .
\end{equation}
The two largest Lyapunov exponents of the system are then $\lambda^{(r)}_1 = \omega_r \gamma$ and $0$, leading to
\begin{equation} \label{eq:asympotic-ee-2d-lyapunov}
    S_A(t) \sim \Lambda_A t \, , \qquad \Lambda_A = \omega_r \gamma \, .
\end{equation}
Comparing Eqs.~\eqref{eq:asympotic-ee-2d}, \eqref{eq:lyapunov-exponents} and \eqref{eq:asympotic-ee-2d-lyapunov}, we see that the asymptotic rate of growth of the entanglement entropy is precisely the value of the positive Lyapunov exponent, which is also equal to $\Lambda_A$.

The coefficent $\Lambda_A$ can also be computed directly from the equation of motion of the normal modes in the case of parametric resonance. A simple model of parametric resonance is provided by the Mathieu equation
\begin{equation}
    \Ddot{x}(t)+\omega^2(t)x(t)=0\,,
\end{equation}
where
\begin{equation} \label{eq:time-dependent-normal-freq}
\omega^2(t)=\omega^2_0 + \alpha \cos(\Omega t)
\end{equation}
is the time-dependent natural frequency, which oscillates around $\omega_0$ with frequency $\Omega=2\omega_0$ and amplitude $\alpha$. It is a well known fact that the \textit{Floquet} exponents $\mu$ can be computed directly from this equation as
\begin{equation} \label{eq:mu}
    \mu=\frac{\alpha}{4 \omega_0}\,.
\end{equation}
In addition, the Floquet exponent is the positive Lyapunov exponent for the stroboscopic evolution in this case. Therefore, in the dynamical Casimir effect, we can write the time-dependent frequency $\Omega_r(t)$ of the resonant mode, given by Eq.~\eqref{eq:omega-d}, in the form \eqref{eq:time-dependent-normal-freq} and compute the Lyapunov exponents directly from Eq.~\eqref{eq:mu}. With $r=(1,1)$, we have
\begin{equation}
    \Omega_r^2=\pi^2 \left[ \frac{1}{L_1^2 ( 1- \epsilon \cos(2\omega_r t))^2}+\frac{1}{L_2^2} \right]\,.
\end{equation}
To first order in $\epsilon$,
\begin{equation}
    \Omega_r^2 = \omega_r^2 + 2 \epsilon \frac{\pi^2}{L_1^2} \cos(2\omega_r t) \, .
\end{equation}
Comparing with the Mathieu equation, we see that $\alpha=2 \epsilon \pi^2/L_1^2$. Substituting into Eq.~\eqref{eq:mu}, setting $\omega_0=\omega_r$, and using Eq.~\eqref{eq:gamma-2p1-d}, we find
\begin{equation}
    \mu= \omega_r \gamma = \Lambda_A \,,
\end{equation}
which agrees with \eqref{eq:lyapunov-exponents}.

In the more explicit derivation of the entropy production rate leading to Eq.~\eqref{eq:asympotic-ee-2d}, we considered a special choice of subsystem $A$. The result remains valid in a much more general context, however. In fact, the arguments leading to the same result in Eq.~\eqref{eq:asympotic-ee-2d-lyapunov} apply to a large class of generic finite-dimensional subsystems, which can be fully characterized from Theorem~\ref{th:SAasymptotics}. In the next section, we provide a detailed characterization of such generic subsystems that applies to any spatial dimension $d\geq 2$.

\begin{figure}[t!]
\centering
  \begin{tikzpicture}
  \draw (0,0) node[inner sep=0pt]{\includegraphics[width=\linewidth]{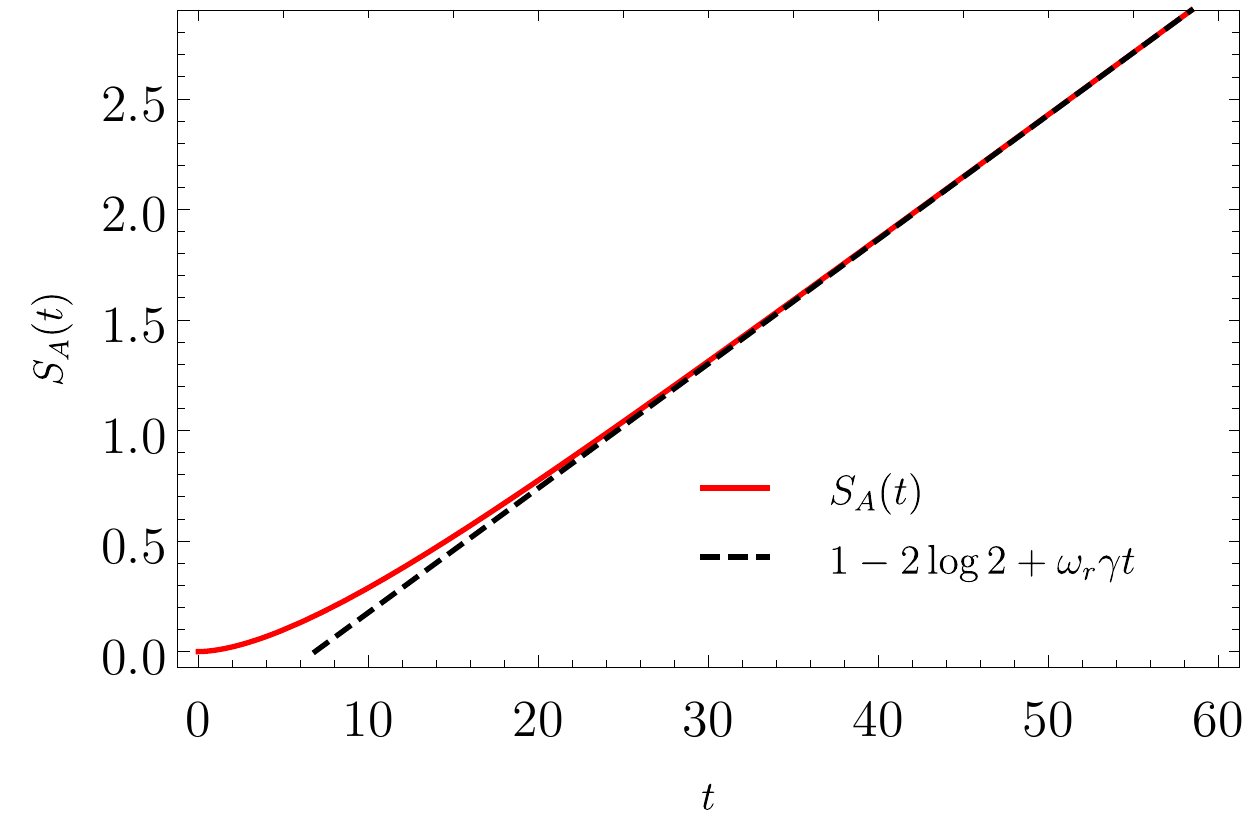}};
  \end{tikzpicture}
  \caption{Comparison of the time-dependent entanglement entropy $S_A(\tau)$ of a resonant mode in $(2+1)$ dimensions with its asymptotic expansion given in Eq.~\eqref{eq:asympotic-ee-2d}.}
  \label{fig4}
\end{figure}

\subsection{Higher dimensions}
For any spatial dimension $d \geq 2$, the dynamical Casimir effect at resonance for the lowest energy mode has the same overall features in the slow-variation approximation, as discussed in Section \ref{sec:d+1-dim}. The time evolution is trivial for all modes, except for the resonant one, which behaves as a single oscillator at parametric resonance. The evolution of the resonant mode is described by the symplectic transformation \eqref{eq:resonant-mode-evolution-d}, with the natural frequency $\omega_r$ and the coefficient $\gamma$ given by Eqs.~\eqref{eq:omega-r} and \eqref{eq:gamma}, respectively. A Lyapunov instability is associated with the resonance. For a generic subsystem, entropy is then produced asymptotically at a constant rate $\Lambda_A=\omega_r \gamma$, which is equal to the positive Lyapunov exponent, $\Lambda_A=\lambda_1^{(r)}$. We now present a simple characterization of such generic subsystems and prove that they satisfy $\Lambda_A=\omega_r \gamma$.

Let $\ell_1,\ell_2\in V^*$ be the eigenvectors of the generator $K^\intercal$ whose eigenvalues $\lambda_1,\lambda_2$ are the nonzero Lyapunov exponents $\lambda_1^{(r)}=-\lambda_2^{(r)}$,
\begin{align}
    \ell_1 = \frac{1}{\sqrt{2}} \begin{bmatrix}
    1 \\ i
    \end{bmatrix} \, , \qquad \ell_2 = \frac{1}{\sqrt{2}} \begin{bmatrix}
    -1 \\ i
    \end{bmatrix} \, .
\end{align}
They form a canonical pair of $V$, $\Omega^{ab} (\ell_1)_a (\ell_2)_b =1$. The remaining eigenvectors of $K^\intercal$ have vanishing eigenvalues. We can construct a linear basis  $\mathcal{D}_L=\{\ell_i\}$ of $V$ composed of canonically conjugate eigenvectors of $K$ that includes $\ell_1$ and $\ell_2$.  The evolution of the basis vectors under the Hamiltonian flow has a simple form:
\begin{align}
    M^\intercal(t)\, \ell_1 &=  e^{\omega_r \gamma t} \ell_1 \, , \nonumber \\
    M^\intercal(t)\, \ell_2 &= e^{-\omega_r \gamma t} \ell_2 \, , \nonumber \\
    M^\intercal(t)\, \ell_i &= \ell_i \, , \quad i > 2 \, .
\end{align}
This basis is adequate for analyzing the expansion of subsystems with time, which determines the asymptotic growth of the R\'enyi and entanglement entropies, as discussed in Section \ref{sec:asymptotic-entanglement}.

Let $A$ be a proper symplectic subspace of $V$ representing a subsystem of interest. Its symplectic complement $B$ is defined as
\begin{equation}
 B = \{ v \in V \mid \Omega_{ab} v^a u^b = 0 , \forall u \in A \} \, .
\end{equation}
We have $V=A\oplus B$, and any vector can be uniquely decomposed as $v=v_A + v_B$, with $v_A \in A$ and $v_B \in B$. The symplectic projection of $v$ to the subsystem $A$ is the component $v_A$ in this decomposition. The decomposition $V=A\oplus B$ induces a unique decomposition $V^*=A^*\oplus B^*$ of the dual phase space based on the isomorphism induced by $\Omega_{ab}^{-1}$, \ie we can define $A^*=\{\Omega^{-1}_{ab}v^b|v^a\in A\}$ and $B^*=\{\Omega^{-1}_{ab}v^b|v^a\in B\}$. We call $A$ (or, equivalently, $A^*$) a generic subsystem if:
\begin{itemize}
    \item[(i)] $\ell_2$ has a nonzero symplectic projection on $A^*$;
    \item[(ii)] $\ell_2$ has a nonzero symplectic projection on $B^*$.
\end{itemize}

The condition (i) implies that there exists a vector $\theta_1 \in A^*$ such that
\begin{equation}
    \Omega^{ab} (\theta_1)_a (\ell_2)_b \neq 0 \, .
\end{equation}
Expanding $\theta_1$ in the basis $\mathcal{D}_L$, $\theta_1 = \sum_{k'}\theta_1^{k'}  \ell_{k'}$, and exploring the fact that $\mathcal{D}_L$ is a Darboux basis, we find that $\theta_1^{1'} \neq 0$, that is, $\theta_1$ has a nonzero expansion coefficient in the first basis vector of $\mathcal{D}_L$. 

Now let $\mathcal{D}_A=\{\theta_r \mid r=1,\dots,2 N_A\}$ be a Darboux basis of the subsystem $A$ that includes $\theta_1$. Expanding its elements in the basis $\mathcal{D}_L$, we can write $\theta_r = \theta_r^{k'}  \ell_{k'}$. By a simple procedure of Gaussian elimination, we can construct a new linear basis $\tilde{\mathcal{D}}_A =\{\tilde{\theta}_r\}$ of $A$ such that $\tilde{\theta}_1 = \theta_1$, and $\tilde{\theta}_r^1 = 0$, for $r \neq 1$. Moreover, it follows from condition (ii) that, for all $r$, $\tilde{\theta}_r^{k'}\neq 0$ for some $k'\neq 2$, otherwise $\ell_2$ would be an element of $A$ and have a vanishing symplectic projection on $B$. In short, the only vector with a contribution from the expanding unstable mode is $\theta_1$, and all other modes have a contribution from some stable mode. As a result, for large times, we have
\begin{align} \label{eq:d-tilde-growth}
    \| M(t)\, \theta_1 \| &\propto  e^{\omega_r \gamma t} \, , \nonumber \\
    \| M(t)\, \theta_i \| &\to \textrm{constant} \, , \quad i \geq 2 \, .
\end{align}

We are interested in computing how the volume of the unit cube with sides $\theta_r$ evolves for large times, $\Vol (M^\intercal(t) \mathcal{D}_A)$. Since the basis $\mathcal{D}_A$ is related to $\tilde{\mathcal{D}}_A$ by a fixed transformation, their volumes are related by a time-independent factor, $\Vol(M(t) \mathcal{D}_A) \propto \Vol(M(t) \tilde{\mathcal{D}}_A)$. But then, from Eq.~\eqref{eq:d-tilde-growth},
\begin{equation}
 \Vol(M^\intercal(t) \mathcal{D}_A) \propto e^{\omega_r \gamma t} \, ,
\end{equation}
and we find for the subsystem exponent:
\begin{align}
    \Lambda_A &= \lim_{t \to \infty} \frac{1}{t} \log \frac{\Vol (M^\intercal(t) \mathcal{D}_A)}{\Vol (\mathcal{D}_A)}\\
    &= \omega_r \gamma \, .
\end{align}
This completes the proof that $\Lambda_A = \omega_r \gamma$ for a generic subsystem satisfying conditions (i) and (ii).

Note that in $(2+1)$ dimensions and higher, the Lyapunov instability is restricted to the resonant mode, so that the condition (ii) is violated for the subsystem composed of the resonant mode only. This is the reason why we had to consider a subsystem mixing more than one mode to obtain a nontrivial behavior in Section \ref{sec:2+1-entropy}.

\section{Discussion}\label{sec:discussion}
We analyzed the production of R\'enyi and entanglement entropies for the dynamical Casimir effect at resonance in arbitrary dimensions. In our settings, a $d$-dimensional parallelepided has one of its perfectly reflecting boundaries oscillating harmonically, while all others are fixed. By setting the frequency of oscillation equal to twice that of some normal mode $r$, the system is set at parametric resonance and particles are continuously produced from the vacuum. We considered the regime where the oscillations of the mirror are small, with a relative amplitude $\epsilon = \Delta L/L \ll 1$. We combined analytical techniques introduced in \cite{Dodonov:1996zz} for the study of pair creation under these circumstances with general symplectic techniques applicable to the analysis of entanglement production in general time-dependent linear systems \cite{Bianchi:2015fra,Bianchi:2017kgb} in order to provide a thorough description of the build up of correlations and entanglement at resonance. 

The system has two characteristic time scales: the period of the resonant mode, $\omega_r^{-1}$, and the time scale $(\omega_r \epsilon)^{-1}$ for which the system starts to depart considerably from the initial configuration at the ground state, which depends on the amplitude of oscillation of the mirror. Since $\epsilon \ll 1$, these scales are well separated. By averaging over the fast oscillations of normal modes, the equations of motion for the mode amplitudes are simplified and can be solved exaclty, as shown in \cite{Dodonov:1996zz}. The cases of $d=1$ and $3$ are treated in \cite{Dodonov:1996zz}, where a Heisenberg representation is used for $d=1$ and a Schr\"odinger representation for $d=3$. We extended the treatment for arbitrary dimensions,  and showed that it can be implemented in the Heisenberg representation in any dimension.

The moving boundary induces a nontrivial evolution of the cavity modes, which in general includes two kinds of interactions: particle creation from the vacuum and mode coupling. The vacuum becomes unstable due to the time-dependent boundary conditions, and the created excitations can be transferred from one mode to another.

For any spatial dimension $d \geq 2$, the intermode interactions are strongly suppressed. This is due to the fact that the energies required to excite nonresonant modes do not match differences between energy levels of the resonant mode, making such transitions unlikely. The amplitude of the resonant mode then increases exponentially, being continuously pumped by the parametric resonance, providing an example of a Floquet instability. In the approximation where the fast oscillations are integrated out, the mode amplitude grows monotonically, and the original Floquet instability then corresponds to a Lyapunov instability. We found that, for any subsystem that intersects both the resonant mode and its complement, R\'enyi and entanglement entropies are produced at a constant rate for large times, $S_A \sim R_A \sim \Lambda_A t$, where the production rate is equal to the positive Lyapunov exponent of the system, as expected from general results of \cite{Bianchi:2017kgb}. It can be explicitly written as $\Lambda_A = \lambda_1^{(r)} = \omega_r \gamma$, with $\omega_r$ and $\gamma$ given by Eqs.~\eqref{eq:omega-r} and \eqref{eq:gamma}. 

An alternate technique for studying a Floquet instability in a system with periodicity $\tau$ consists of looking at its stroboscopic evolution. That is, if $M(t)$ is its Hamiltonian flow, one focus on the discrete series of instants $M(n \tau)$, $n \in \mathbb{Z}$. An effective, time-independent Hamiltonian can then be defined in terms of $\tau^{-1} \log M(\tau)$ such that at full periods its evolution agrees with that of the original system. This effective Hamiltonian displays a Lyapunov instability associated with the original Floquet instability. Production of entropy in systems with Floquet instabilities was studied in detail in \cite{Bianchi:2017kgb} within the stroboscopic approach. Applying it to the dynamical Casimir effect at resonance for $d \geq 2$, we found the same results for the entanglement entropy production as in the approximation introduced in \cite{Dodonov:1996zz}, where the fast oscillations are integrated out. Our result suggests that, for systems satisfying a slow-variation condition, $M(\tau) \sim \mathbb{1}$, the technique of averaging out the fast oscillations provides a good approximation for computing the Lyapunov exponents of the effective Hamiltonian associated with the stroboscopic dynamics of the system.

The case of one spatial dimension displays special features that lead to a different asymptotic behavior for the R\'enyi and entanglement entropies. In constrast to what happens in higher dimensions, for $d=1$, the energy required to excite any normal mode is a multiple of the difference $\hbar \omega_r$ between energy levels of the lowest energy mode, which we set at resonance. All modes are then strongly coupled, and excitations produced at the resonant mode can transition to other modes. The energy injected in the field by the moving mirror thereby spreads throughout an infinite number of modes. As a result, two new effects take place. First, the number of particles in the resonant mode, instead of growing exponentially, increases only linearly with time, as first recognized in \cite{Dodonov:1996zz}. In addition, the resonant mode is strongly entangled with the nonresonant modes. We found that the associated R\'enyi and entanglement entropies of the resonant mode display a logarithmic growth, $S_A \sim 1/2 \log \tau$, where $\tau=\epsilon \omega t/2$ is a dimensionless time variable. It is interesting that, even in the presence of parametric resonance, the system does not display features expected for unstable systems, i.e., a linear production of entropy. Instead, its behavior is qualitatively similar to that of metastable Hamiltonians discussed in \cite{Hackl:2017ndi}, except for the fact that the logarithmic function is here multiplied by a factor of $1/2$, while for metastable systems the prefactor is always an integer.

The techniques here applied for the analysis of entropy production in the dynamical Casimir effect at resonance can also be applied to other mirror trajectories. For finite cavities, the adaptation of the procedure followed in this work should be straightforward. A natural question for further developments is how to implement our techniques for the dynamical Casimir effect with a single mirror.  As is well known, one of the main motivations for the study of the dynamical Casimir effect is that it can mimic the Hawking effect \cite{Davies:1976hi}. The application of symplectic techniques for the study of entropy production in general field theories is discussed in \cite{Bianchi:2017kgb}, and we hope our approach can be extended along such lines to the analysis of the time evolution of the entanglement entropy in models simulating aspects of particle creation in the Hawking effect \cite{Carlitz:1986nh,Good:2013lca}. 

\begin{acknowledgments}
We thank Eugenio Bianchi and Rodolfo R. Soldati for inspiring discussions during the development of this project. We thank Eugenio Bianchi for sharing his notes on Floquet instabilities and for many invaluable discussions during the conception and initial stages of this work. IR acknowledges support from CAPES and CNPq. LH thanks the members of the physics department at the Universidade Federal de Minas Gerais for their hospitality during his visit in the summer of 2018. LH acknowledges support through the Brazil-U.S. Exchange Program of the American Physical Society which made this visit possible. LH is funded by the the Max Planck Harvard Research Center for Quantum Optics and supported by the Deutsche Forschungsgemeinschaft (DFG, German Research Foundation) under Germany’s Excellence Strategy – EXC-2111 – 39081486. NY acknowledges support from CNPq, Brazil, under the grant PQ 306744/2018-0, and from the Programa Institucional de Aux\'ilio \`a Pesquisa de Docentes Rec\'em-Contratados, PRPq/UFMG.
\end{acknowledgments}

\bibliographystyle{apsrev4-1}
\bibliography{DCE}

\end{document}